\documentclass[a4paper,11pt]{article}

\usepackage{graphicx}

\usepackage{fullpage}
\usepackage{libertine}
\usepackage{color}

\usepackage[ruled,vlined]{algorithm2e}
\SetArgSty{textrm} 
\renewcommand{\algorithmcfname}{Mechanism}

\usepackage{amsmath,amsfonts,amssymb,amsthm}

\usepackage[breaklinks=true]{hyperref}
\usepackage[svgnames]{xcolor}
\usepackage[capitalise,nameinlink]{cleveref}
\hypersetup{colorlinks={true},linkcolor={DarkBlue},citecolor=[named]{DarkGreen}}

\usepackage{natbib}
\usepackage{authblk}

\usepackage{subcaption}

\usepackage{tikz}  
\usetikzlibrary{arrows}
\usetikzlibrary{patterns,snakes}
\usetikzlibrary{decorations.shapes}
\tikzstyle{overbrace text style}=[font=\tiny, above, pos=.5, yshift=5pt]
\tikzstyle{overbrace style}=[decorate,decoration={brace,raise=5pt,amplitude=3pt}]
\usetikzlibrary{shapes.geometric}

\newtheorem{theorem}{Theorem}[section]

\newtheorem{lemma}[theorem]{Lemma}

\newtheorem{property}[theorem]{Property}

\theoremstyle{definition}
\newtheorem*{comment*}{Comment}

\newcommand{\cost}{\text{cost}}

\newcommand{\SC}{\text{SC}}

\newcommand{\bw}{\mathbf{w}}
\newcommand{\bx}{\mathbf{x}}
\newcommand{\by}{\mathbf{y}}
\newcommand{\bo}{\mathbf{o}}
\newcommand{\med}{\mathtt{med}}

\title{\bf Distributed Agent-Constrained Truthful Facility Location}
\author[1]{Argyrios Deligkas}
\affil[1]{Royal Holloway University of London, UK}
\author[2]{Panagiotis Kanellopoulos}
\author[2]{Alexandros A. Voudouris}
\affil[2]{University of Essex, UK}

\date{}

\begin{document}

\allowdisplaybreaks

\maketitle

\begin{abstract}
We study a distributed facility location problem in which a set of agents, each with a private position on the real line, is partitioned into a collection of fixed, disjoint groups. The goal is to open $k$ facilities at locations chosen from the set of positions reported by the agents. This decision is made by mechanisms that operate in two phases. In Phase 1, each group selects the position of one of its agents to serve as the group’s representative location. In Phase 2, $k$ representatives are chosen as facility locations.
Once the facility locations are determined, each agent incurs an individual cost, defined either as the sum of its distances to all facilities (sum-variant) or as the distance to its farthest facility (max-variant). We focus on the class of strategyproof mechanisms, which preclude the agents from benefiting through strategic misreporting, and establish tight bounds on the approximation ratio with respect to the social cost (the total individual agent cost) in both variants.
\end{abstract}

\section{Introduction} \label{sec:intro}
Imagine you are the newly elected mayor of Streetville, a quiet village in Metric County, where cars~are banned, everyone lives along the same infinite street, and residents enjoy playing a game called Find the Liar. During your campaign, you enthusiastically promised the construction of several new public facilities --- a park, a hospital, and a school. You also pledged that citizens would have a say in the decision and that every facility would be ``within a minute’s walk.'' Now it is time to finally deliver. But how can you actually make these decisions to satisfy the residents as much as possible? One idea is to ask everyone for their ideal location and then build each facility so as to minimize the total walking distance for all citizens --- after all, they don’t seem to like spending too much time outside. But, what if people lie about their ideal location to reduce their own walking distance? You turn to Professor Choice for advice, who immediately recognizes this as a mechanism design problem and suggests studying the {\em truthful facility location problem}~\citep{fl-survey}.

Many different variants of the truthful facility location problem have been studied, depending on the assumptions made, for example, about the number of facilities and the preferences of the agents over them. When a single facility is to be opened, it is natural to assume that every agent would prefer to be as close to it as possible. In the case of multiple facilities, however, the preferences of agents might depend on the types of services that the facilities provide. A large part of the literature has considered settings where the facilities are homogeneous, and thus agents aim to be close to one of them~\citep{procaccia09approximate,Lu2010two-facility,fotakis2014two}. However, when the facilities are heterogeneous, agents might prefer to be close to more than one facility and suffer a cost that is a function of their distances to all facilities. In particular, two main individual cost functions have been proposed: one that considers the distance of an agent to all facilities (known as the {\em sum-variant}), and another that considers the distance to the farthest facility (known as the {\em max-variant}).

Another crucial aspect of the problem concerns the feasible locations where facilities can be opened. The earliest works in this area considered the continuous case, in which facilities can be opened at any point on the line. In many real-world applications, however, the setting is discrete, and facility placement is restricted to existing agent locations due to physical, legal, or budgetary constraints, or even due to pre-existing infrastructure. There are also several social choice settings, such as {\em peer selection} and {\em sortition}, where the goal is for the agents to choose a subset of themselves that forms a decision-making committee. Motivated by such applications, many different location-constrained models have been proposed and studied in recent years, in which facilities can be opened only at predetermined candidate locations~\citep{feldman2016voting,Tang2020candidate,kanellopoulos2025}, or at the locations reported by the agents~\citep{Deligkas2025agent}.

There are also several applications, such as federated or regional planning, where collecting input from all agents simultaneously may be logistically too difficult or even prohibited, for instance when agents are partitioned into disjoint administrative regions or belong to different sub-networks. In such scenarios, the decision is typically computed in a distributed manner. Agents are partitioned into groups (such as focus groups or subcommittees), and each group locally proposes a single location that is representative of the positions of all agents within the group. The outcome is then computed based on the group proposals; for example, facilities may be opened at a subset of these representative locations. This setting is known as the {\em distributed facility location} model and was first considered by \citet{FilosRatsikasKVZ24}, who focused on the case of a single facility. More recently, \citet{Xu2025distributed} studied a slightly different distributed model for two facilities, in which each group proposes two locations.

\subsection{Our Contribution}
In this paper, we consider distributed facility location problems with $n$ {\em agents} that are partitioned into $m$ disjoint {\em groups}, and there are $k \leq m$ facilities to open. The agents have {\em private positions} on the line of real numbers and report them to an algorithm (or, {\em mechanism}) which decides the locations of the facilities by implementing the following two phases. 
In Phase 1, for each group, given the positions of the agents therein, the mechanism chooses one of these positions as the representative location of the whole group. 
In Phase 2, it chooses where to open the facilities from the set of representative locations of all groups. Once the facility locations have been determined, each agent suffers an individual cost, which is either her total distance to the facilities ({\em sum-variant}), or her distance to the farthest facility ({\em max-variant}). Naturally, an agent acts as cost-minimizer and can potentially misreport her true position if this can decrease her individual cost. We focus on the design of deterministic {\em strategyproof} mechanisms that make decisions to eliminate the strategic behavior of the agents and at the same time achieve a good approximation of the minimum possible {\em social cost} (defined as the total individual cost of the agents). 

We paint an almost complete picture of the best possible approximation ratio that can be achieved by strategyproof mechanisms. In particular, for $k=2$ facilities, we show a tight bound of $1+\sqrt{2} \approx 2.414$ in the sum-variant, and a tight bound of $9/2 = 4.5$ in the max-variant. The $1+\sqrt{2}$ bound in the sum-variant turns out to not only be the best possible among strategyproof mechanisms, but also among unrestricted algorithms that ignore the strategic behavior of the agents. In the max-variant, a slightly better bound of $4$ can be achieved for such unrestricted algorithms. For $k \geq 3$ facilities, we show that the best possible approximation ratio of strategyproof mechanisms is between $3-2/k$ and $3+2/k$ in the sum-variant (leaving a small gap that diminishes as $k$ increases), and between $2k$ and $2(k+1)$ in the max-variant. An overview of these results is given in Table~\ref{tab:overview}.

\renewcommand{\arraystretch}{1.3}
\begin{table}[t]
\centering
\begin{tabular}{c|c|c}
             &   $k=2$ 		 & $k \geq 3$ \\ \hline 
Sum-variant  &   $1+\sqrt{2}$ & $[3-2/k, 3+2/k]$ \\ \hline
Max-variant	 &	 $9/2$	      &	$[2k, 2(k+1)]$ \\ \hline
\end{tabular}
\caption{Overview of our bounds on the approximation ratio of deterministic strategyproof mechanisms for the different variants (sum- and max-) and the number $k$ of available facilities. Single numbers for $k=2$ indicate tight bounds. }
\label{tab:overview}
\end{table}

Our upper bounds (positive results) follow by strategyproof mechanisms that make choices according to ordered statistics. That is, such mechanisms select the group representative locations according to the ordering of the positions reported by the agents therein, and then decide the final facility location based on the ordering of the representative locations. While this is a common mechanism design paradigm in various truthful facility location problems, in our distributed setting, especially for the case of $k=2$ facilities, it turns out that multiple such mechanisms must be appropriately combined depending on the number of groups, in order to create mechanisms that achieve the best possible approximation ratio bounds. Our lower bounds (negative results) follow by, in cases quite elaborate, constructions that exploit a locality property of distributed mechanisms which dictates that the same representative location must be chosen for any group with the same number of agents that have the same positions, as well as the fact that all representative locations must be part of the solution when $k=m$.

\subsection{Related Work}
Over the last decades, truthful facility location problems have been studied extensively within the area of {\em approximate mechanism design without money}, with early works establishing foundational models and approximation results mainly for settings where the facilities can be opened anywhere on the line~\citep{procaccia09approximate,Lu2010two-facility,fotakis2014two}; see also the survey of \citet{fl-survey}. 
An ever-growing body of research has investigated different truthful $k$-facility location models with additional feasibility constraints, most notably when facilities must be chosen from a finite set of candidate locations. This constraint has been explored under various objectives and information models, including the sum- and max-variants that we also consider here~\citep{Sui2015constrained,feldman2016voting,Walsh2021limited,Tang2020candidate,lotfi2024max,Zhao2023constrained,serafino2016,kanellopoulos2023discrete,kanellopoulos2025,gai2024mixed}. One crucial aspect in all these works is that the set of candidate locations is fixed and known a priori. In contrast, in our model, the candidate locations are the positions of the agents and are thus dynamically determined once the agents report them, as in the recent work of \citet{Deligkas2025agent}. Besides models where the facilities are desirable, there has also been significant work on obnoxious facility location with candidate locations, where agents prefer facilities to be opened as far away as possible~\citep{ZhaoLNF24,GaiLW22,kanellopoulos2025obnoxious}. Another related constraint, which can be thought of as orthogonal to having candidate locations, is requiring a minimum distance between the facilities, aiming for better distribution and fairness~\citep{Xu2021minimum,duan2021minimum}.

Going beyond centralized decision-making and inspired by similar work in voting, \citet{FilosRatsikasKVZ24} were the first to study distributed facility location models, focusing exclusively on the case of one facility. In their model, as in our setting, the agents are partitioned into groups, and each group independently proposes a representative location (which can be any point on the line), and then one of those is chosen to open the facility. Recently, \citet{Xu2025distributed} considered a setting with two facilities. In their model, the facilities can only be opened on fixed candidate locations, and each group proposes two representative locations. This is in contrast to our model where the facilities can be opened on agent locations and each group only proposes one such location. Complementary to these distributed models, another line of recent work has considered group facility location models, where, the agents are partitioned into groups that potentially affect their individual costs, but the decision is centralized~\citep{ZhouLC2022group,LiLC2024group-obnoxious,LiPWZ2025group}. 

Our setting is also closely related to {\em utilitarian voting} where agents have (metric) preferences over a set of candidates, typically expressed via ordinal rankings, and the objective is to select one or more candidates so as to minimize the social cost or maximize the social welfare. We refer to the survey of~\citet{distortion-survey} for a comprehensive overview of this literature and the notion of {\em distortion}. Within this framework, the most related works are those focusing on distributed variants, where, similarly to our facility location model and to the ones mentioned above, agents are grouped together in disjoint districts and communication among agents in different districts is prohibited~\citep{Filos-RatsikasMV20distributed,AnshelevichFV2022distributed,Voudouris2023tight,Voudouris2025obnoxious,Abam2025randomized-distributed}. \citet{AmanatidisAJV2025group} recently studied another related, centralized voting model with unknown groups of agents to evaluate how well outcomes approximate the optimal solution in terms of group-fair objectives.

Another line of research considers, similarly to our work, {\em peer selection} settings in which agents themselves are the candidates, aiming to model representative selection and social choice among individuals \citep{Cheng2017people,Cheng2018representative,Cembrano2025peer}. A related concept is that of {\em $\alpha$-decisiveness} in metric voting, used to capture structured preferences where $\alpha \in [0,1]$ is a factor indicating how close an agent is to all the candidates compared to how close she is to her top-ranked candidate, with $\alpha=0$ implying that an agent actually coincides with her top-ranked candidate in the metric space~\citep{Anshelevich2018approximating,Anshelevich2017randomized,gkatzelis2020resolving,kempe2022veto}.
Finally, a less related setting where agents act as candidates is that of {\em impartial selection}, where a winner or a small set of winners must be selected from a group of agents who nominate one another, under the constraint that no agent can influence their own chance of being selected~\citep{AlonFPT2011sum,FischerK2015impartial,CembranoFK2023improved}.

\section{The Model} \label{sec:model}
We consider a distributed facility location model with a set $N$ of $n \geq 2$ agents and $k \geq 2$ facilities. Each agent $i$ has a position $x_i \in \mathbb{R}$ on the line of real numbers. Let $\bx = (x_i)_{i \in N}$ be the {\em position profile} consisting of all agent positions. The agents are partitioned into $m \geq k$ disjoint fixed groups that are known a priori. Let $G$ be the set of all groups, and denote by $n_g =|g|$ the size of group $g\in G$. 
So, an instance of our problem can be fully described by a tuple $I = (\bx, G, k)$ which defines the positions of the agents, the groups of agents, and the number of facilities to be opened. 

Our goal is to build all $k$ facilities at {\em distinct} locations chosen from the set of positions reported by the agents; note that two facilities are allowed to be opened at the same location if there are two agents with that position. In particular, a solution is a $k$-tuple $\by = (y_\ell)_{\ell \in [k]}$ indicating the location $y_\ell$ of facility $\ell \in [k]$, such that $\by \subseteq \bx$. In our distributed setting, an algorithm $\mathcal{A}$ works by implementing the following two phases:
\begin{itemize}
    \item Phase 1: For each group $g \in G$, the algorithm decides the position of an agent $r_g \in g$ as the representative location for the whole group. 
    \item Phase 2: The algorithm chooses $k$ out of the $m$ group representative locations where the $k$ facilities are opened. That is, it outputs a solution $\by \subseteq \bigcup_{g \in G} \{x_{r_g}\}$.
\end{itemize}
For any instance $I$, let $\mathcal{A}(I)$ be the solution computed by $\mathcal{A}$. We will typically use $\bw_I$ (or simply $\bw$ if $I$ is clear from context) to denote the solution of an algorithm in our proofs. 

Let $d(i,y) = |x_i-y|$ be the distance of agent $i$ from any point $y$ on the line.
Given a solution $\by = (y_\ell)_{\ell \in [k]}$, the {\em individual cost} of an agent $i$ in instance $I$ is defined as 
\begin{itemize}
    \item {\bf (sum-variant)} the total distance from all facilities: $\cost^{\text{sum}}_i(\by|I) = \sum_{\ell \in [k]} d(i,y_\ell)$;
    \item {\bf (max-variant)} the distance from the farthest facility: $\cost^\text{max}_i(\by|I) = \max_{\ell \in [k]} d(i,y_\ell)$.
\end{itemize}
When the variant is clear from context, we will simply write $\cost_i(\bw|I)$ to denote the individual cost of agent $i$ for $\by$. 
Since the groups might in general be of different sizes, we define the {\em social cost} of a solution $\by$ in instance $I$ as the average over groups of the average individual cost of the agents within the groups, that is 
\begin{align*}
    \SC(\by|I) = \frac{1}{m}\sum_{g \in G} \frac{1}{n_g}\sum_{i \in g} \cost_i(\by|I).
\end{align*}
We will typically use $\bo_I$ (or simply $\bo$ if $I$ is clear from context) to denote a solution that minimizes the social cost.
Observe that, if we were allowed to open all facilities at the same point, then, to minimize the social cost, that point would be the {\em weighted median} of the agents. In particular, let $\pi$ be the ranking of all agents according to their positions (breaking ties arbitrarily). Each agent $i$ in a group $g$ has a weight $\mathtt{v}_i = 1/n_g$. The weighted median $i^*$ is an agent such that $\sum_{i: \pi(i) < \pi(i^*)} \mathtt{v}_i < m/2$ and $\mathtt{v}_{i^*} + \sum_{i: \pi(i) < \pi(i^*)} \mathtt{v}_i \geq m/2$. When the groups are {\em symmetric} (i.e., they all have the same size $n/m$), our social cost definition reduces to the {\em average individual agent cost}, and the weighted median agent reduces to the median agent. For simplicity, if the groups are symmetric, we will calculate the total individual agent cost as the social cost. 

The {\em approximation ratio} $\rho_k(\mathcal{A})$ of an algorithm $\mathcal{A}$ is the worst-case ratio (over all instances with a given number $k$ of facilities) between the social cost of the computed solution and the minimum social cost achieved by any feasible solution, that is
\begin{align*}
    \rho_k(\mathcal{A}) = \sup_{I = (\bx, G, k)} \frac{\SC(\mathcal{A}(I)|I)}{\min_{\by \subseteq \bx: |\by| = k} \SC(\by|I)}.
\end{align*}

In many cases, it is important to avoid strategic manipulations by the agents who aim to minimize their individual costs (rather than the social cost). In particular, an algorithm $\mathcal{A}$ (also referred to as {\em mechanism} in this case) is {\em strategyproof} if, for any agent $i$ and two instances $I=((x_i,\bx_{-i})_,G,k)$ and $J = ((x_i',\bx_{-i}), G, k)$ that differ only on the position of agent $i$, we have that 
$\cost_i(\mathcal{A}(I)|I) \leq \cost_i(\mathcal{A}(J)|I)$. Our goal is to design {\em unrestricted algorithms} and {\em strategyproof mechanisms} with an as small as possible approximation ratio.

\section{A Class of Strategyproof Mechanisms for Two Facilities} \label{sec:k=2:strategyproof-upper}
We start by presenting a class of mechanisms for $k=2$ facilities. We will argue that all members of this class are strategyproof mechanisms, and show a general bound on the approximation ratio they achieve for both the sum- and the max-variant. For any parameters $\theta, \ell, r \in (0,1]$ such that $\lceil \ell m\rceil <\lceil r m \rceil$, mechanism $\mathcal{M}_m(\theta, \ell, r)$ works as follows: In Phase 1, for each group $g \in G$, it selects the position of the $\lceil \theta n_g \rceil$-leftmost agent therein as $g$'s representative $r_g$. In Phase 2, the mechanism opens the facilities at the $\lceil \ell m\rceil$- and $\lceil r m\rceil$-leftmost representative locations. See Mechanism~\ref{mech:k=2:mechanism} for a description using pseudocode. Note that the constraint $\lceil \ell m\rceil < \lceil r m\rceil$ that we impose on the parameters guarantees that we never co-locate the two facilities unless there are at least two agents there.

\SetCommentSty{mycommfont}
\begin{algorithm}[h]
\SetNoFillComment
\caption{$\mathcal{M}_m(\theta, \ell, r)$}
\label{mech:k=2:mechanism}
{\bf Input:} Instance $I = (\bx, G, 2)$, parameters $\theta, \ell, r \in (0,1]$ such that $\lceil \ell m\rceil <\lceil r m \rceil$\;
{\bf Output:} Facility locations $\bw = (w_1,w_2)$\;
$\mathtt{Rep} \gets \varnothing$\;
\For{each group $g \in G$}
{
    $i_g \gets \lceil \theta n_g \rceil$-leftmost agent in $g$\;
    $r_g \gets x_{i_g}$\;
    $\mathtt{Rep} \gets \mathtt{Rep} \cup \{r_g\}$\;
}
$w_1 \gets \lceil \ell m\rceil$-leftmost location in $\mathtt{Rep}$\;
$w_2 \gets \lceil r m\rceil$-leftmost location in $\mathtt{Rep}$\;
\end{algorithm}

We first argue that mechanism $\mathcal{M}_m(\theta, \ell, r)$ is strategyproof, for any (valid) parameters and any of the two variant settings we consider (sum and max).

\begin{theorem}\label{thm:k=2:general-sp}
For any $m\geq 2$ and $\theta, \ell, r \in (0, 1]$, mechanism $\mathcal{M}_m(\theta, \ell, r)$ is strategyproof.
\end{theorem}

\begin{proof}
Consider a group $g$ with representative $r_g$ and an agent $i \in g$ with true position $x_i$. Suppose that $i$ reports a different position $x_i'$ and let $\bw' = (w_1',w_2')$ be the resulting solution. If the representative of $g$ does not change after $i$ misreports, then the overall solution does not change, and thus $i$ suffers the same cost as before. If the representative of $g$ does change to $r_g'$, then it must be the case that either $x_i \leq r_g < r_g'$ or $r_g' < r_g \leq x_i$. We now switch between different cases depending on the relative locations of $w_1, w_2, r_g, x_i$ and $r_g'$. 

\medskip
\noindent
{\bf Case 1: $r_g > w_2$.}
Clearly, the solution does not change if $r'_g > r_g$. Otherwise, if $r_g' < r_g \leq x_i$, since $w_1 \leq w_2 < r_g \leq x_i$, either the solution does not change, or it changes to $\bw'$ with $w'_1\leq w_1$ and $w'_2\leq w_2$. In any case, the cost of $i$ does not decrease.

\medskip
\noindent
{\bf Case 2: $r_g < w_1$.}
This is symmetric to Case 1.

\medskip
\noindent
{\bf Case 3: $r_g \in (w_1, w_2)$.}
If $x_i \leq r_g < r'_g$, then the new solution $\bw'$ is such that $w_2' \geq w_2$ and $w_1'=w_1$. Otherwise, if $r'_g < r_g \leq x_i$, the new solution $\bw'$ is such that $w_1'\leq w_1$ and $w_2' = w_2$. In both cases, the cost of $i$ may only increase.

\medskip
\noindent
{\bf Case 4: $r_g \in \{w_1,w_2\}$.}
Without loss of generality, suppose that $r_g = w_1$. If $r'_g<r_g$, we have $w'_1\leq w_1$ and $w'_2=w_2$. As $x_i\geq r_g$, such a change may only increase the cost of agent $i$. Otherwise, if $r_g < r'_g$, we have $w'_1\geq w_1$ and $w'_2\geq w_2$. Since $x_i\leq r_g$, such a change may again only increase the cost of agent $i$.
\end{proof}

Next, we focus on the approximation ratio that mechanism $\mathcal{M}_m(\theta, \ell, r)$ can achieve as a function of the parameters $\theta$, $\ell$ and $r$, as well as the number $n$ of agents and the number $m$ of groups. This will allow us in the next section to identify the best possible combinations of the parameters for the different settings we consider. Before analyzing the approximation ratio, we argue that any instance can be transformed into another with symmetric groups (where each group $g$ has size $n_g = n/m$) such that the mechanism outputs the same solution and the position of the weighted median in the original instance is the same as the position of the median agent in the new instance. Using this, we can without loss of generality focus on bounding the approximation ratio of the mechanism on instances with symmetric groups. 

\begin{lemma}\label{lem:k=2-wlog-property}
    For any instance $I$ with $m$ groups, there is another instance $J$ with $m$ symmetric groups such that the representative location chosen by $\mathcal{M}_m(\theta, \ell, r)$ for a group is the same in $I$ and $J$, and the position of the weighted median agent in $I$ is the same as the position of the median agent in $J$.
\end{lemma}

\begin{proof}
Let $G_I$ and $G_J$ be the sets of groups in instance $I$ and instance $J$, respectively. 
For every group $g \in G_I$ of size $n_g$ in instance $I$, we create a group $g' \in G_J$ in instance $J$ of size $\prod_{q \in G_I} n_q$ as follows: For each agent $i \in g$ that is positioned at $x_i$, $g'$ contains $\prod_{q \in G_I \setminus\{g\}} n_q$ agents at $x_i$. 

It is now not hard to see that the representative of group $g'$ in $J$ is the same as the representative of group $g$ in $I$: The representative of $g'$ is the position of the $\lceil \theta \prod_{q \in G_I} n_q \rceil$-leftmost agent in $g'$. Since the position $x_i$ of each agent $i\in g$ is repeated $\prod_{q \in G_I \setminus\{g\}} n_q$ times in $g'$, the representative location of $g'$ is the same as the position of the $\left\lceil \frac{\theta \prod_{q \in G_I} n_q}{\prod_{q \in G_I \setminus\{g\}} n_q} \right\rceil = \lceil \theta n_g \rceil$-leftmost agent in $g$, which is exactly the position chosen as the representative of $g$ by the mechanism. For the same reason, the position of the median agent in $J$ is the same as the position of the weighted median agent in $I$. 
\end{proof}

We are now ready to bound the approximation ratio achieved by the mechanism. 

\begin{theorem}\label{thm:k=2:general-upper}
Let $Y$ be an indicator variable that is $1$ if the setting is the sum-variant, and $0$ if the setting is the max-variant. 
For any $m\geq 2$  and $\theta, \ell, r\in (0, 1]$, mechanism $\mathcal{M}_m(\theta, \ell, r)$ has approximation ratio at most 
\begin{align*}
\max\bigg\{&\frac{n}{(m+1-\lceil \ell m \rceil)(\frac{n}{m}+1- \lceil \frac{\theta n}{m}\rceil)}-1, \frac{n}{(1+Y)(m+1-\lceil rm\rceil)(\frac{n}{m}+1-\lceil\frac{\theta n}{m}\rceil)}, \\
&\frac{n}{\lceil rm \rceil\lceil \frac{\theta n}{m}\rceil} -1, \frac{n}{(1+Y)\lceil \ell m\rceil\lceil\frac{\theta n}{m}\rceil} \bigg\}.  
\end{align*}
\end{theorem}

\begin{proof}
Let $\bw = (w_1,w_2)$ be the solution computed by the mechanism, $\bo = (o_1,o_2)$ be an optimal solution, and $\med$ be the position of the overall median agent. We distinguish between three cases depending on the relative positions of $w_1$, $w_2$ and $\med$: $\med \leq w_1$, $\med \geq w_2$, and $\med \in (w_1, w_2)$. We partition the agents into the following three sets: $L = \{i: x_i\leq w_1\}$, $M = \{i: w_1 < x_i < w_2\}$, and $R = \{i: x_i \geq w_2\}$.

\medskip
\noindent 
{\bf Case 1: $\med \leq w_1$.} 
We first derive upper bounds on the individual costs of the agents depending on which set they belong to. 
\begin{itemize}
    \item For any agent $i \in L$, using the triangle inequality, 
    \begin{align*}
        \cost_i(\bw) 
        &= Y \cdot d(i,w_1) + d(i,w_2) \\
        &\leq Y \cdot \bigg( d(i,\med) + d(\med,w_1) \bigg) + d(i,\med) + d(\med,w_1) + d(w_1,w_2) \\
        &= (1+Y) d(i,\med) + (1+Y) d(\med,w_1) + d(w_1,w_2).
    \end{align*}

    \item For any agent $i \in M$, 
    \begin{align*}
        \cost_i(\bw) = Y\cdot \min\{d(i,w_1),d(i,w_2)\} + \max\{d(i,w_1),d(i,w_2)\} \leq  d(w_1,w_2).
    \end{align*}

    \item For any agent $i \in R$, 
    \begin{align*}
        \cost_i(\bw) = d(i,w_1) + Y\cdot d(i,w_2) &= (1+Y) d(i,w_2) + d(w_1,w_2).
    \end{align*}
\end{itemize}
Therefore, the social cost of the solution $\bw$ computed by the mechanism is
\begin{align}
    \SC(\bw) &= \sum_{i \in N} \cost_i(\bw) \nonumber \\
    &= \sum_{i \in L} \cost_i(\bw) + \sum_{i \in M} \cost_i(\bw) + \sum_{i \in R} \cost_i(\bw) \nonumber  \\
    &\leq \sum_{i \in L} \bigg( (1+Y) d(i,\med) + (1+Y) d(\med,w_1) + d(w_1,w_2) \bigg) + \sum_{i \in M} d(w_1,w_2) \nonumber  \\
    &\quad + \sum_{i \in R} \bigg( (1+Y) d(i,w_2) + d(w_1,w_2) \bigg) \nonumber  \\
    &= (1+Y) \sum_{i \in L} d(i,\med) + (1+Y) \sum_{i \in R} d(i,w_2) + (1+Y)|L|\cdot d(\med,w_1) + n \cdot d(w_1,w_2) \nonumber  \\
    &= (1+Y) \sum_{i \in L} d(i,\med) + (1+Y) \sum_{i \in R} d(i,w_2) \nonumber  \\
    &\quad + (1+Y) \bigg(n-(|M|+|R|)\bigg)\cdot d(\med,w_1) + n \cdot d(w_1,w_2). \label{eq:k=2:general-upper:SCbw}
\end{align}
For the optimal solution, since $\med$ minimizes the total distance from all other agents, the cost of the optimal solution is at least as much as (virtually) opening both facilities there. Consequently, we have
\begin{align*}
    \SC(\bo) &\geq (1+Y)\sum_{i \in N} d(i,\med) \\
    &= (1+Y)\sum_{i \in L} d(i,\med) + (1+Y)\sum_{i \in M} d(i,\med) +  (1+Y)\sum_{i \in R} d(i,\med).
\end{align*}
Since $\med \leq w_1$, we have that $d(i,\med) \geq d(\med,w_1)$ for any agent $i \in M$, and $d(i,\med) = d(i,w_2) + d(\med,w_1) + d(w_1,w_2)$ for any agent $i \in R$. Hence,
\begin{align}
    \SC(\bo) &\geq (1+Y)\sum_{i \in L} d(i,\med) + (1+Y)|M|\cdot d(\med,w_1) \nonumber \\
    &\quad + (1+Y)\sum_{i \in R} d(i,w_2) + (1+Y)|R| \cdot \bigg( d(\med,w_1) + d(w_1,w_2) \bigg)  \nonumber  \\
    &=  (1+Y)\sum_{i \in L} d(i,\med) +  (1+Y) \sum_{i \in R} d(i,w_2) \nonumber  \\
    &\quad + (1+Y)\bigg( |M| + |R|\bigg) d(\med,w_1) + (1+Y)|R| \cdot d(w_1,w_2). \label{eq:k=2:general-upper:SCbo}
\end{align}
By rearranging terms, we have
\begin{align*}
    &(1+Y) \sum_{i \in L} d(i,\med) +  (1+Y) \sum_{i \in R} d(i,w_2) \\
    &\leq \SC(\bo) - (1+Y)\bigg( |M| + |R|\bigg) d(\med,w_1) - (1+Y)|R| \cdot d(w_1,w_2).
\end{align*}
Hence, by \eqref{eq:k=2:general-upper:SCbw}, we can upper-bound $\SC(\bw)$ as follows:
\begin{align*}
     \SC(\bw) 
     &\leq \SC(\bo) + (1+Y)\bigg(n-2(|M|+|R|)\bigg)\cdot d(\med,w_1) + \bigg(n - (1+Y)|R|\bigg)\cdot d(w_1,w_2).
\end{align*}
By \eqref{eq:k=2:general-upper:SCbo}, we also have that
\begin{align*}
    \SC(\bo) \geq (1+Y)\bigg( |M| + |R|\bigg) \cdot d(\med,w_1) + (1+Y)|R| \cdot d(w_1,w_2).
\end{align*}
So, the approximation ratio is 
\begin{align*}
    \frac{\SC(\bw)}{\SC(\bo)} 
    &\leq 1 + \frac{(1+Y)\bigg(n-2(|M|+|R|)\bigg)\cdot d(\med,w_1) + \bigg(n - (1+Y)|R|\bigg)\cdot d(w_1,w_2)}{(1+Y)\bigg( |M| + |R|\bigg) \cdot d(\med,w_1) + (1+Y)|R| \cdot d(w_1,w_2)} \\
    &\leq 1 + \max\left\{ \frac{n-2(|M|+|R|)}{|M| + |R|},  \frac{n - (1+Y)|R|}{(1+Y)|R|}\right\} \\
    &= \max\left\{ \frac{n}{|M| + |R|}-1,  \frac{n}{(1+Y)|R|}\right\}.
\end{align*}
Observe that the last expression is decreasing in terms of $|M|$ and in terms of $|R|$, and thus its maximum possible value is attained for the minimum possible values of $|M|$ and $|R|$. By the definition of $w_1$ and $w_2$, we have that 
$$|R|\geq (m+1-\lceil r m \rceil) \left( \frac{n}{m} + 1 - \left\lceil \frac{\theta n}{m} \right\rceil \right)$$ 
and 
$$|M| + |R| \geq (m+1-\lceil \ell m \rceil) \left( \frac{n}{m} + 1 - \left\lceil \frac{\theta n}{m} \right\rceil \right).$$
Hence, the approximation ratio is at most
\begin{align*}
 \frac{\SC(\bw)}{\SC(\bo)} &\leq \max\left\{ \frac{n}{|M| + |R|}-1,  \frac{n}{(1+Y)|R|}\right\} \\
 &\leq \max\left\{\frac{n}{(m+1-\lceil \ell m \rceil)(\frac{n}{m}+1- \lceil \frac{\theta n}{m}\rceil)}-1, \frac{n}{(1+Y)(m+1-\lceil rm\rceil)(\frac{n}{m}+1-\lceil\frac{\theta n}{m}\rceil)}\right\}.
\end{align*}

\medskip
\noindent 
{\bf Case 2: $w_2 \leq \med$.}
This case is symmetric to Case 1 in the sense that the roles of sets $L$ and $R$ are swapped. 
This leads to an approximation ratio 
\begin{align*}
    \frac{\SC(\bw)}{\SC(\bo)} 
    &\leq \max\left\{ \frac{n}{|M| + |L|}-1,  \frac{n}{(1+Y)|L|}\right\}.
\end{align*}
By the definitions of $w_1$ and $w_2$, we now have that
\begin{align*}
  |L| \geq \lceil \ell m \rceil \left\lceil \frac{\theta n}{m} \right\rceil
\end{align*}
and 
\begin{align*}
  |M| + |L| \geq  \lceil r m \rceil \left\lceil \frac{\theta n}{m} \right\rceil.
\end{align*}
Therefore, the approximation ratio is 
\begin{align*}
    \frac{\SC(\bw)}{\SC(\bo)} 
    &\leq \max\left\{ \frac{n}{\lceil r m \rceil \left\lceil \frac{\theta n}{m} \right\rceil}-1,  \frac{n}{(1+Y)\lceil \ell m \rceil \left\lceil \frac{\theta n}{m}\right\rceil}\right\}.
\end{align*}

\medskip
\noindent 
{\bf Case 3: $w_1 < \med < w_2$.}
We have:
\begin{itemize}
    \item For any agent $i \in L$, using the triangle inequality, 
    \begin{align*}
        \cost_i(\bw) 
        &= Y \cdot d(i,w_1) + d(i,w_2) = (1+Y) d(i,w_1) +  d(w_1,w_2).
    \end{align*}

    \item For any agent $i \in M$, 
    \begin{align*}
        \cost_i(\bw) = Y\cdot \min\{d(i,w_1),d(i,w_2)\} + \max\{d(i,w_1),d(i,w_2)\} \leq  d(w_1,w_2).
    \end{align*}

    \item For any agent $i \in R$, 
    \begin{align*}
        \cost_i(\bw) = d(i,w_1) + Y\cdot d(i,w_2) &= (1+Y) d(i,w_2) + d(w_1,w_2).
    \end{align*}
\end{itemize}
Therefore, we have the following upper bound on the social cost of $\bw$:
\begin{align}
    \SC(\bw) 
    &= \sum_{i \in N} \cost_i(\bw) \nonumber \\
    &\leq (1+Y) \sum_{i \in L} d(i,w_1) + (1+Y) \sum_{i \in R} d(i,w_2) + n\cdot d(w_1,w_2). \label{eq:k=2:general-upper:SCbw:Case3}
\end{align}
For the optimal solution, we have
\begin{align}
    \SC(\bo) 
    &\geq (1+Y) \sum_{i \in N} d(i,\med) \nonumber \\
    &\geq  (1+Y) \sum_{i \in L} d(i,\med) + (1+Y)\sum_{i \in R} d(i,\med) \nonumber \\
    &= (1+Y) \sum_{i \in L} \bigg( d(i,w_1) + d(w_1,\med) \bigg) + (1+Y) \sum_{i \in R} \bigg( d(i,w_2) + d(\med,w_2) \bigg) \nonumber \\
    &\geq (1+Y) \sum_{i \in L} d(i,w_1) + (1+Y) \sum_{i \in R} d(i,w_2) +  (1+Y) \min\{|L|,|R|\} \cdot d(w_1,w_2), \label{eq:k=2:general-upper:SCbo:Case3}
\end{align}
where the last inequality follows by the fact that $d(w_1,\med) + d(\med,w_2) = d(w_1,w_2)$. 
By rearranging terms, we have that
\begin{align*}
(1+Y) \sum_{i \in L} d(i,w_1) + (1+Y) \sum_{i \in R} d(i,w_2) \leq \SC(\bo) - (1+Y) \min\{|L|,|R|\} \cdot d(w_1,w_2)
\end{align*}
and thus \eqref{eq:k=2:general-upper:SCbw:Case3} implies that
\begin{align*}
    \SC(\bw) \leq \SC(\bo) + \bigg(n - (1+Y) \min\{|L|,|R|\} \bigg) \cdot d(w_1,w_2).
\end{align*}
By \eqref{eq:k=2:general-upper:SCbo:Case3}, we also have that
\begin{align*}
    \SC(\bo) &\geq (1+Y)\min\{|L|,|R|\} \cdot d(w_1,w_2).
\end{align*}
So, the approximation ratio is
\begin{align*}
    \frac{\SC(\bw)}{\SC(\bo)} 
    &\leq \frac{\SC(\bo) + \bigg(n - (1+Y)\min\{|L|,|R|\} \bigg) \cdot d(w_1,w_2)}{\SC(\bo)} \\
    &\leq 1 + \frac{n - (1+Y)\min\{|L|,|R|\} }{(1+Y)\min\{|L|,|R|\}} \\
    &= \frac{n}{(1+Y)\min\{|L|,|R|\}}.
\end{align*}
By the definition of $w_1$ and $w_2$, similarly to the previous two cases, we have that 
$$|L| \geq \lceil \ell m \rceil \left\lceil \frac{\theta n}{m} \right\rceil$$ 
and 
$$|R| \geq (m+1-\lceil r m \rceil) \left( \frac{n}{m} + 1 - \left\lceil \frac{\theta n}{m} \right\rceil \right).$$ Therefore, the approximation ratio is 
\begin{align*}
    \frac{\SC(\bw)}{\SC(\bo)} \leq \max\left\{\frac{n}{(1+Y)\lceil \ell m \rceil \lceil \frac{\theta n}{m} \rceil}, \frac{n}{(1+Y)(m+1-\lceil r m \rceil) \left( \frac{n}{m} + 1 - \lceil \frac{\theta n}{m} \rceil \right)}\right\}.
\end{align*}
This completes the proof.
\end{proof}

\section{Tight Strategyproof Bounds for Two Facilities}\label{sec:k=2:sp-bounds}
Using Theorem~\ref{thm:k=2:general-upper}, we can now design strategyproof mechanisms that achieve low approximation ratios for both the sum- and the max-variant settings by appropriately choosing the parameters $\theta$, $\ell$ and $r$ depending on the number $m$ of groups. Since $\mathcal{M}_m(\theta,\ell,r)$ is strategyproof for any $\theta, \ell, r \in (0,1]$, and the number $m$ of groups is a priori known, running different versions of $\mathcal{M}_m(\theta,\ell,r)$ depending on the value of $m$ leads to a strategyproof mechanism. Taking advantage of this, we will show an upper bound of $1+\sqrt{2}\approx 2.414$ in the sum-variant, and an upper bound of $9/2=4.5$ in the max-variant. We will also prove that these are the best possible bounds over all strategyproof mechanisms. 

\subsection{Sum-variant}\label{sec:k=2:sp:sum}

We start with the sum-variant. 

\begin{theorem}\label{thm:sum:k=2:upper}
In the sum-variant, for $k=2$ facilities, there exists a strategyproof mechanism with approximation ratio at most $1+\sqrt{2}$. 
\end{theorem}

\begin{proof}
For any $m \neq 3, 5$, we run $\mathcal{M}_m(1/2,\sqrt{2}-1,2-\sqrt{2})$. By Theorem~\ref{thm:k=2:general-upper}, with $Y=1$ (as we are in the sum-variant setting), using the fact that $xm \leq \lceil x m \rceil \leq xm+1$ for any $x \in \{\ell,r\}$, we can first simplify the upper bound, and then substitute these values for the parameters $\theta$, $\ell$ and $r$, to get an approximation ratio of at most
\begin{align*}
    &\max\left\{\frac{1}{(1-\ell)(1-\theta)}-1, \frac{1}{2(1-r)(1-\theta)},\frac{1}{r\theta}-1,\frac{1}{2\ell \theta}\right\} \\
    &=\left\{\frac{2}{2-\sqrt{2}}-1,\frac{1}{\sqrt{2}-1}, \frac{2}{2-\sqrt{2}}-1, \frac{1}{\sqrt{2}-1}\right\}  = 1+\sqrt{2}.
\end{align*}
Note that for $m\in\{3,5\}$, $\lceil (\sqrt{2}-1)m \rceil = \lceil (2-\sqrt{2})m \rceil$, which means that these parameter values cannot be used for the mechanism. So, we have to treat each of these cases separately using a different combination of parameters.    

For $m=3$, we run $M_3(1/3,2/3,1)$. First, observe that for these parameters we have $\lceil 3\ell\rceil < \lceil 3r\rceil$. By Theorem~\ref{thm:k=2:general-upper}, substituting $\theta=1/3$, $\ell = 2/3$, $r=1$, and using the fact that $n/9\leq \lceil n/9 \rceil \leq n/9+1$, the approximation ratio is at most
\begin{align*}
&\max\bigg\{\frac{n}{2(\frac{n}{3}+1- \lceil \frac{n}{9}\rceil)}-1, \frac{n}{2(\frac{n}{3}+1-\lceil\frac{n}{9}\rceil)}, \frac{n}{3\lceil \frac{n}{9}\rceil} -1, \frac{n}{4\lceil\frac{n}{9}\rceil} \bigg\} \\
&= \max\left\{\frac{9}{4}, \frac{9}{4}, 2, \frac{9}{4}\right\} = \frac94.
\end{align*}

Finally, for $m=5$, we run $M_5(2/5, 3/5,4/5)$. Again observe that for these values we have $\lceil 5\ell\rceil < \lceil 5r\rceil$. 
Theorem~\ref{thm:k=2:general-upper}, substituting $\theta=2/5$, $\ell = 3/5$, $r=4/5$, and using the fact that $2n/25 \leq \lceil 2n/25 \rceil \leq 2n/25+1$, the approximation ratio is at most
\begin{align*}
&\max\bigg\{\frac{n}{3(\frac{n}{5}+1- \lceil \frac{2n}{25}\rceil)}-1, \frac{n}{4(\frac{n}{5}+1-\lceil\frac{2n}{25}\rceil)},\frac{n}{ 4 \lceil \frac{2n}{25}\rceil} -1, \frac{n}{6\lceil\frac{2n}{25}\rceil} \bigg\} \\
&=\max\left\{\frac{16}{9}, \frac{25}{12}, \frac{17}{8}, \frac{25}{12}\right\} = \frac{25}{12}.
\end{align*}
We conclude that, for any $m$, the approximation ratio is at most $1+\sqrt{2}$. 
\end{proof}

We will now show that $1+\sqrt{2}$ is the best possible approximation ratio we can hope for over any algorithm, even non-strategyproof ones.

\begin{theorem}
In the sum-variant, for $k=2$ facilities, the approximation ratio of any algorithm is at least $1+\sqrt{2}-\varepsilon$, for any $\varepsilon > 0$.
\end{theorem}

\begin{proof}
We will consider instances with symmetric groups consisting of many agents. For simplicity, we assume that the total volume of a group is $1$.  We also allow the number of groups to be irrational numbers (so that we can avoid the use of ceilings; this will lead to a lower bound of $1+\sqrt{2}$ rather than $1+\sqrt{2}-\varepsilon$). For any $\alpha \in [0,1]$, let $g_\alpha$ be a group such that an $\alpha$-fraction of the agents therein are at $0$, while the remaining $(1-\alpha)$-fraction of agents are at $1$. Without loss of generality, assume that the representative of group $g_{1/2}$ is $1$. 
To show the statement, we will need the following property, which will be proven later. 

\begin{property}\label{property:k=2:sum:main}
Consider any algorithm with approximation ratio strictly smaller than $1+\sqrt{2}$, and let $x$ and $y$ be two large numbers. 
The algorithm must output the solution $(1,1)$ when there are $x$ groups with representative $0$ and $y$ groups with representative $1$ such that $y/x = \sqrt{2}$.
\end{property}

Now, suppose that there is an algorithm with approximation ratio strictly smaller than $1+\sqrt{2}$. 
Consider an instance with $x$ copies of $g_1$ (each with representative $0$) and $y$ copies of $g_{1/2}$ (each with representative $1$), where $x$ and $y$ are large integers such that $y/x = \sqrt{2}$.
By Property \ref{property:k=2:sum:main}, the algorithm must output the solution $(1,1)$. 
Since 
\begin{align*}
    \SC(1,1) = \left(x+\frac{y}{2}\right) \cdot 2 = y + 2x
\end{align*}
and 
\begin{align*}
    \SC(0,0) = \frac{y}{2} \cdot 2 = y,
\end{align*}
the approximation ratio is $1+ 2x/y= 1+\sqrt{2}$, a contradiction. 

The rest of the proof is dedicated to proving Property~\ref{property:k=2:sum:main}. We start by showing that any algorithm with an approximation ratio strictly smaller than $1+\sqrt{2}$ must choose $0$ as the representative of $g_{1/\sqrt{2}}$. Consider an instance consisting of two copies of this group. Clearly, there are two possible solutions depending on the representative of the group: $(0,0)$ if the representative is $0$, or $(1,1)$ if the representative is $1$. Since 
$$\SC(1,1) = \frac{2}{\sqrt{2}}$$
and 
$$\SC(0,0) = 2 \cdot \left(1-\frac{1}{\sqrt{2}}\right),$$ 
the approximation ratio would be $1+\sqrt{2}$ if $0$ is not the representative of the group, contradicting that the approximation is strictly smaller than that much. 

Next, we claim the following property. 

\begin{property} \label{property:k=2:sum:solution}
If $0$ is chosen as the representative of $g_\alpha$ for some $\alpha > 1/2$, then the algorithm must output solution $(1,1)$ when there are $x$ groups with representative $0$ and $y$ groups with representative $1$ such that $\lambda = y/x =  2(1+\sqrt{2})\alpha-1$.  
\end{property}

\begin{proof}
Suppose otherwise, and consider an instance consisting of $x$ copies of $g_\alpha$ that have $0$ as their representative, and $y$ copies of $g_0$ that have $1$ as their representative. Then, 
$$\SC(0,0) = \left( x(1-\alpha) + y \right) \cdot 2,$$
$$\SC(0,1) = x+y$$ 
and
$$\SC(1,1) = x\alpha \cdot 2.$$
Hence, the approximation is at least
$$\min\left\{ \frac{x(1-\alpha)+y}{x\alpha}, \frac{x+y}{2x\alpha} \right\} = \frac{x+y}{2x\alpha} = \frac{1+\lambda}{2\alpha} = 1+\sqrt{2},$$
if $(1,1)$ is not the solution, contradicting that the algorithm has an approximation ratio strictly smaller than $1+\sqrt{2}$.
\end{proof}

Finally, we claim the following property. 

\begin{property} \label{property:k=2:sum:group}
If the algorithm outputs solution $(1,1)$ when there are $x$ groups with representative $0$ and $y$ groups with representative $1$, then it must choose $0$ as the representative of group $g_\alpha$ for $\alpha = \frac{(1+\sqrt{2})\lambda-1}{(2+\sqrt{2})\lambda}$, where $\lambda = y/x$.
\end{property}

\begin{proof}
Assume otherwise that the representative of $g_\alpha$ is $1$, and consider an instance with $x$ copies of $g_1$ that have representative $0$ and $y$ copies of $g_\alpha$ that have representative $1$ by assumption. Then, the solution is $(1,1)$. Since
$$\SC(1,1) = (x + \alpha y) \cdot 2$$
and 
$$\SC(0,0) = (1-\alpha)y \cdot 2$$
the approximation ratio is at least
$$\frac{x+\alpha y}{(1-\alpha) y} = \frac{1+\alpha \lambda}{(1-\alpha)\lambda} = 1+\sqrt{2},$$
which contradicts that the algorithm has an approximation ratio strictly smaller than that much.
\end{proof}

Using Properties \ref{property:k=2:sum:solution} and \ref{property:k=2:sum:group},  we can argue about the solutions in certain instances and the representatives of particular groups. That is, since there is an initial $\alpha_0 = 1/\sqrt{2}$ such that $g_{\alpha_0}$ must have representative $0$ to achieve approximation ratio at most $1+\sqrt{2}$, we obtain a collection of values for $\lambda$ and $\alpha$ that meet the conditions of the properties. In particular, these properties define a recurrence as follows:
Let $t \geq 1$ be a time step, then, by Property~\ref{property:k=2:sum:solution}, we have that
\begin{align*}
    \lambda_t = 2(1+\sqrt{2}) \alpha_{t-1} - 1, 
\end{align*}
and, by Property~\ref{property:k=2:sum:group}, we have that
\begin{align*}
    \alpha_t = \frac{(1+\sqrt{2})\lambda_t-1}{(2+\sqrt{2})\lambda_t}.
\end{align*}
Taking $t$ to infinity, the stable solution $(\lambda^*,\alpha^*)$ of this recurrence is when 
\begin{align*}
    \lambda^* = 2(1+\sqrt{2}) \alpha^* - 1 \text{ \ \ and \ \ }  \alpha^* = \frac{(1+\sqrt{2})\lambda^*-1}{(2+\sqrt{2})\lambda^*}.
\end{align*}
Substituting the second equality to the first one, we obtain a quadratic equation for $\lambda$ that has solution $\lambda^*=\sqrt{2}$; substituting this value in the second equality, we also obtain that $\alpha^* = 1/2$. For these values, Property~\ref{property:k=2:sum:solution} implies Property~\ref{property:k=2:sum:main}, and the proof of the lower bound is now complete. 
\end{proof}

\subsection{Max-variant} \label{sec:k=2:sp:max}
We now turn to the max-variant. Using Theorem~\ref{thm:k=2:general-upper} again, we design a strategyproof mechanism with approximation ratio at most $9/2 = 4.5$.

\begin{theorem}\label{thm:max:k=2:upper}
In the max-variant, for $k=2$ facilities, there exists a strategyproof mechanism with approximation ratio at most $9/2$. 
\end{theorem}

\begin{proof}
First, by Theorem~\ref{thm:k=2:general-upper}, with $Y=0$ (as we are in the max-variant setting), for any valid parameters $\ell$ and $r$ that satisfy the inequality $\lceil \ell m\rceil < \lceil r m \rceil$, the upper bound on the approximation ratio achieved by mechanism $\mathcal{M}_m(\theta,\ell,r)$ can be simplified to 
\begin{align*}
&{\textstyle
\max\bigg\{\frac{n}{(m+1-\lceil \ell m \rceil)(\frac{n}{m}+1- \lceil \frac{\theta n}{m}\rceil)}-1, \frac{n}{(m+1-\lceil rm\rceil)(\frac{n}{m}+1-\lceil\frac{\theta n}{m}\rceil)}, \frac{n}{\lceil rm \rceil\lceil \frac{\theta n}{m}\rceil} -1, \frac{n}{\lceil \ell m\rceil\lceil\frac{\theta n}{m}\rceil} \bigg\}} \\
&= \max\bigg\{\frac{n}{(m+1-\lceil rm\rceil)(\frac{n}{m}+1-\lceil\frac{\theta n}{m}\rceil)}, \frac{n}{\lceil \ell m\rceil\lceil\frac{\theta n}{m}\rceil} \bigg\}.  
\end{align*}

For any even number $m \geq 2$ of groups, we run mechanism $\mathcal{M}_m(1/2, 1/2, 1/2+\varepsilon)$. This mechanism opens the facilities at the two median representative locations in Phase 2. The approximation ratio is at most 
\begin{align*}
    &\max\bigg\{\frac{n}{(m+1-\lceil (\frac12+\varepsilon)m\rceil)(\frac{n}{m}+1-\lceil\frac{n}{2m}\rceil)}, \frac{n}{\lceil \frac{m}{2}\rceil\lceil\frac{n}{2m}\rceil} \bigg\} \\
    &\leq \max\bigg\{\frac{2n}{ m (\frac{n}{m}+1-\lceil\frac{n}{2m}\rceil)}, \frac{2n}{m \lceil\frac{n}{2m}\rceil} \bigg\}  = 4.
\end{align*}

For any odd number $m \geq 3$ of groups, we run mechanism $\mathcal{M}_m(\frac{m-1}{2m}, \frac{m+1}{2m}, \frac{m+1}{2m}+\varepsilon)$. This mechanism opens the facilities at the median representative and the representative directly to its right. The approximation ratio is now at most
\begin{align*}
    &\max\bigg\{\frac{n}{(m+1-\lceil (\frac{m+1}{2m}+\varepsilon) m\rceil)(\frac{n}{m}+1-\lceil\frac{n(m-1)}{2m^2}\rceil)}, \frac{n}{\lceil \frac{m+1}{2m} m \rceil\lceil\frac{n(m-1)}{2m^2}\rceil} \bigg\} \\
    &\leq \max\bigg\{\frac{2n}{(m-1)(\frac{n}{m}+1-\lceil \frac{n(m-1)}{2m^2}\rceil)}, \frac{2n}{(m+1)\lceil \frac{n(m-1)}{2m^2}\rceil}\bigg\}\leq \frac{4m^2}{m^2-1} \leq \frac92.
\end{align*}
The proof is now complete.
\end{proof}

We complement the above positive result with a tight lower bound of $9/2$ on the approximation ratio of any strategyproof mechanism for two facilities in the max-variant. Our construction uses instances with three groups of agents. We should remark at this point that, from the proof of Theorem~\ref{thm:max:k=2:upper}, it follows that we can achieve a better upper bound of $4$ when the number of groups is even (in the next section we will actually argue that $4$ is the best possible bound over all unrestricted algorithms in the max-variant), so the worst-case necessarily arises when there is an odd number of groups. To prove the lower bound, we will need the following property of strategyproof mechanisms.  

\begin{lemma} \label{lem:k=2:max:sp-continuity-property}
For any real numbers $a$ and $b$ such that $a < b$, and any positive integers $x$ and $y$, let $g_{a,b}(x,y)$ be a group with $x$ agents at $a$ and $y$ agents at $b$.
\begin{itemize}
    \item If a strategyproof mechanism chooses $a$ as the representative of $g_{a,b}(x,y)$, then it must choose $a'$ as the representative of $g_{a',b}(x,y)$ for any $a' \in [a,b)$, and $a$ as the representative of $g_{a,b'}(x,y)$ for any $b' \in (a,b]$. 
    \item If a strategyproof mechanism chooses $b$ as the representative of $g_{a,b}(x,y)$, then it must choose $b'$ as the representative of $g_{a,b'}(x,y)$ for any $b' \in (a,b]$, and $b$ as the representative of $g_{a',b}(x,y)$ for any $a' \in [a,b)$.
\end{itemize}

\end{lemma}

\begin{proof}
Since the two cases are symmetric, it suffices to prove the first one. For simplicity, let $a = 0$, $b = 1$, and drop $(x,y)$ from notation. So, we have a strategyproof mechanism that chooses $0$ as the representative of group $g_{0,1}$. 

\medskip
\noindent 
{\bf First part.}
Suppose towards a contradiction that, for some real number $a_1 \in [0,1/2)$, the mechanism chooses $1$ as the representative of group $g_{a_1,1}$. Clearly, the solution must be $(0,0)$ for any instance that consists only of copies of $g_{0,1}$, and the solution must be $(1,1)$ for any instance that consists only of copies of $g_{a_1,1}$.
Consider now an instance that consists of two copies of $g_{a_1,1}$. One by one, the $2x$ agents at $a_1$ move to $0$. Each such move leads to a new instance and the mechanism, since it is strategyproof, must necessarily include $1$ in the solution as, otherwise, the deviating agent would decrease her cost from $1-a_1> 1/2$ to at most $a_1 < 1/2$. At the end, all of these agents are at $0$, and we have an instance consisting of two copies of $g_{0,1}$, but the solution is either $(0,1)$ or $(1,1)$, contradicting that the solution must be $(0,0)$.

For any integer $j \geq 2$, we can repeat the same argument and show by contradiction that the mechanism must choose $a_j$ as the representative of group $g_{a_j,1}$, where $a_j \in [a_{j-1}, \frac{1+a_{j-1}}{2})$. Doing so for every possible integer proves the statement.

\medskip
\noindent 
{\bf Second part.}
Consider an instance that consists of a copy of $g_{0,1}$ and a copy of group $g_{b'}$, where all agents are at $b'$. The representatives of the two groups are $0$ and $b'$, and thus the solution for this instance is $(0,b')$. One by one, the $y$ agents at $1$ in $g_{0,1}$ move to $b'$. Each such move leads to a new instance, and the mechanism must still choose $0$ as the group's representative as, otherwise, the solution would change to $(b',b')$ and thus the deviating agent would decrease her cost from $1$ to $1-b'$, contradicting that the mechanism is strategyproof.
\end{proof}

We are now ready to prove the lower bound. 

\begin{theorem} \label{thm:k=2:max:lower}
In the max-variant, for $k=2$ facilities, the approximation ratio of any strategyproof mechanism is at least $9/2$.
\end{theorem}

\begin{proof}
Suppose there is a strategyproof mechanism with approximation ratio strictly smaller than $9/2$. We will show that this is impossible using a construction with instances of three groups, each consisting of three agents.
For any real number $a$, let $g_a$ be the group in which all three agents are at $a$. 
For any real numbers $a, b$ such that $a < b$, and integers $x, y$ such that $x>0$, $y>0$ and $x+y=3$, 
let $g_{a,b}(x,y)$ be a group with $x$ agents at $a$ and $y$ agents at $b$.
Clearly, the representative of $g_a$ is $a$, and the representative of $g_{a,b}(x,y)$ is in $\{a,b\}$. 
We consider the following cases. 

\medskip
\noindent
{\bf Case 1: The representative of group $g_{0,1}(1,2)$ is $0$.}
Consider first the case where the representatives of the three groups are $0$, $1$ and $1$. 
Suppose that the mechanism includes $0$ in the solution. 
Then, the instance might consist of a copy of $g_{0,1}(1,2)$ and two copies of $g_1$.
Since $\SC(0,1) = 9$ and $\SC(1,1)=1$, the approximation ratio of the mechanism is at least $9$ in this case.
Consequently, the mechanism must output the solution $(1,1)$ when the group representatives are $0$, $1$ and $1$. 

Consider now the case where the group representatives are $0$, $1/2+\varepsilon$ and $1$, for some infinitesimal $\varepsilon$. 
Suppose that the mechanism includes $0$ in the solution.
Then, the instance might be such that the second group is a copy of $g_{1/2+\varepsilon}$. 
One by one, these agents move to $1$, transforming the group into a copy of $g_1$, and thus the group representatives become $0$, $1$ and $1$. 
By the above, the mechanism then outputs the solution $(1,1)$, which means that the last moving agent decreases her cost from $1/2+\varepsilon$ to $1/2-\varepsilon$, contradicting that the mechanism is strategyproof.  
Consequently, the mechanism must output the solution $(1/2+\varepsilon, 1)$ when the group representatives are $0$, $1/2+\varepsilon$ and $1$. 

Next, consider an instance with group representatives $0$, $1/2+\varepsilon$ and $1$, where the first group is a copy of $g_0$. As argued above, the solution is $(1/2+\varepsilon, 1)$. One by one, the agents at $0$ in the first group move to $1/2+\varepsilon$, transforming it into a copy of $g_{1/2+\varepsilon}$, and thus leading to an instance with group representatives $1/2+\varepsilon$, $1/2+\varepsilon$ and $1$. Since the mechanism is strategyproof, it must still include $1$ in the solution of any intermediate instance as the agents deviate as, otherwise, the cost of such an agent would decrease from $1$ to at most $1/2+\varepsilon$. So, we have established that the mechanism must output the solution $(1/2+\varepsilon, 1)$ when the group representatives are $1/2+\varepsilon$, $1/2+\varepsilon$ and $1$.

Finally, consider an instance consisting of two copies of $g_{1/2+\varepsilon,1}(1,2)$ and a copy of $g_1$. 
By Lemma~\ref{lem:k=2:max:sp-continuity-property}, since the mechanism is strategyproof and the representative of $g_{0,1}(1,2)$ is $0$, the representative of $g_{1/2+\varepsilon,1}(1,2)$ is $1/2+\varepsilon$. Hence, the group representatives are $1/2+\varepsilon$, $1/2+\varepsilon$ and $1$, implying by the above that the solution must be $(1/2+\varepsilon, 1)$. Since $\SC(1/2+\varepsilon, 1) = 9(1/2-\varepsilon)$ and $\SC(1,1) = 2(1/2-\varepsilon)$, the approximation ratio is $9/2$.

\medskip
\noindent
{\bf Case 2: The representative of group $g_{0,1}(2,1)$ is $1$.}
This case can be handled similarly to Case 1 by replacing $g_{0,1}(1,2)$ with $g_{0,1}(2,1)$, swapping references to locations $0$ and $1$, and changing $1/2+\varepsilon$ to $1/2-\varepsilon$.

\medskip
\noindent
{\bf Case 3: The representative of group $g_{0,1}(1,2)$ is $1$ and the representative of group $g_{0,1}(2,1)$ is $0$.}
Consider the case where the group representatives are $0$, $1$ and $1$. 
Suppose that the mechanism includes $0$ in the solution. 
Then, the instance might consist of a copy of $g_{0,1}(2,1)$ and two copies of $g_1$. 
Since $\SC(0,1) = 9$ and $\SC(1,1)=2$, the approximation ratio of the mechanism is $9/2$ in this case. 
Consequently, it must be the case that the mechanism outputs $(1,1)$ when the group representatives are $0$, $1$ and $1$. 

Consider now the case where the group representatives are $0$, $1/2+\varepsilon$ and $1$. 
Suppose that the mechanism includes $0$ in the solution. 
Then, the instance might be such that the second group is a copy of $g_{1/2+\varepsilon}$. 
One by one, these agents move to $1$, transforming the group into a copy of $g_1$, and thus the group representatives are $0$, $1$ and $1$. 
By the above, the mechanism then outputs the solution $(1,1)$, which means that the last moving agent decreases her cost from $1/2+\varepsilon$ to $1/2-\varepsilon$, contradicting that the mechanism is strategyproof. 
Consequently, we have that the mechanism outputs the solution $(1/2+\varepsilon, 1)$ when the group representatives are $0$, $1/2+\varepsilon$ and $1$.

Next, consider an instance with group representatives $0$, $1/2+\varepsilon$ and $1$, where the first group is a copy of $g_0$. 
As argued above, the solution is $(1/2+\varepsilon, 1)$. One by one, the agents at $0$ in the first group move to $1/2+\varepsilon$, transforming it into a copy of $g_{1/2+\varepsilon}$, and thus leading to an instance with representatives $1/2+\varepsilon$, $1/2+\varepsilon$ and $1$. Since the mechanism is strategyproof, it must include $1$ in the outcome of any intermediate instance as the agents deviate as, otherwise, the cost of any such agent would decrease from $1$ to at most $1/2+\varepsilon$. Hence, we have that the computed solution is $(1/2+\varepsilon,1)$ when the group representatives are $1/2+\varepsilon$, $1/2+\varepsilon$ and $1$. 

Finally, consider an instance with two copies of $g_{1/2+\varepsilon}$ and one copy of $g_{1/2+\varepsilon,1}(1,2)$; by Lemma~\ref{lem:k=2:max:sp-continuity-property}, the representatives are $1/2+\varepsilon$, $1/2+\varepsilon$ and $1$. Since $\SC(1/2+\varepsilon,1) = 9(1/2-\varepsilon)$ and $\SC(1/2+\varepsilon,1/2+\varepsilon) = 2(1/2-\varepsilon)$, the approximation ratio is at least $9/2$.
\end{proof}

\section{Unrestricted Algorithms for Two Facilities}\label{sec:k=2:unrestricted-upper}
Here we focus on unrestricted algorithms. In the sum-variant, as already mentioned in Section~\ref{sec:k=2:sp:sum}, the lower bound of $1+\sqrt{2}$ holds even for non-strategyproof mechanisms, implying that $1+\sqrt{2}$ is the best possible approximation ratio over any algorithm in this setting. In the max-variant, we will show a tight bound of $4$. As briefly discussed in Section~\ref{sec:k=2:sp:max}, the proof of Theorem~\ref{thm:max:k=2:upper} implies that our strategyproof mechanism achieves this upper bound when the number $m$ of groups is even. The lower bound is implied by Theorem~\ref{thm:kgeq3:max:lower} for $k=2$ which is presented in Section~\ref{sec:kgeq3:sp} and holds for any unrestricted algorithm. So, it suffices to focus on the case of odd $m$, and show an upper bound of $4$; actually, we will show an improved bound of $7/2$.

We consider the following algorithm $\mathcal{A}_{\max}^*$ (see Algorithm~\ref{mech:k=2:unrestricted-algorithm} for a description using pseudocode): 
In Phase 1, for each group, the algorithm picks the leftmost median agent as the group's representative. 
In Phase 2, it opens the facilities at the median representative and the representative closest to it. 

\renewcommand{\algorithmcfname}{Algorithm}
\begin{algorithm}[h]
\SetNoFillComment
\caption{$\mathcal{A}_{\max}^*$}
\label{mech:k=2:unrestricted-algorithm}
{\bf Input:} Instance $I = (\bx, G, 2)$\;
{\bf Output:} Facility locations $\bw = (w_1,w_2)$\;
$\mathtt{Rep} \gets \varnothing$\;
\For{each group $g \in G$}
{
    $i_g \gets \lceil n_g/2 \rceil$-leftmost agent in $g$\;
    $r_g \gets x_{i_g}$\;
    $\mathtt{Rep} \gets \mathtt{Rep} \cup \{r_g\}$\;
}
$w_1 \gets \lceil m/2\rceil$-leftmost location in $\mathtt{Rep}$\;
$w_2 \gets \arg\min_{r \in \mathtt{Rep}\setminus\{w_1\}} d(r,w_1)$\;
\end{algorithm}

Before bounding its approximation ratio, we present a simple instance with three groups showing that $\mathcal{A}_{\max}^*$ is not strategyproof. 
The first group consists of an agent at $0$ and an agent at $0.1$, the second group consists of two agents at $0.51$, and the third group consists of two agents at $1$. Since each group is of size $2$, its representative is the position of the leftmost agent therein; so, the representatives are $0$, $0.51$ and $1$. The two facilities are opened at $0.51$ (the median representative) and at $1$ (the representative closest to the median one). Observe that the agent at $0$ in the first group (whose position is selected as that group's representative) has cost $1$. By misreporting her position as $0.1$, the representative location of the first group becomes $0.1$ which is closer to $0.51$ than $1$ is, and thus the facilities are opened at $0.1$ and $0.51$. This decreases the cost of the deviating agent from $1$ to $0.51$, proving that the algorithm is not strategyproof.

We now focus on showing an upper bound of $7/2$ on the approximation ratio of the algorithm for odd $m$. Our analysis will follow along the lines of that in Section~\ref{sec:k=2:strategyproof-upper}. Note that, since the mechanism chooses the representatives of the groups according to the ordering of the agents therein, it suffices to consider instances with symmetric groups (in particular, it is not hard to observe that the same property as the one stated for mechanisms $\mathcal{M}_m(\theta,\ell,r)$ in Lemma~\ref{lem:k=2-wlog-property} holds for $\mathcal{A}_{\max}^*$ as well). 

\begin{theorem}\label{thm:max-unrestricted-upper}
In the max variant, for $k=2$ facilities and odd number $m$ of groups, the approximation ratio of algorithm $\mathcal{A}_{\max}^*$ is at most $7/2$.
\end{theorem}

\begin{proof}
Let $w_m$ be the median representative, and denote by $w_\ell$ and $w_r$ the representatives directly to the left and to the right of $w_m$, respectively. Clearly, one of $w_\ell$ and $w_r$ is the closest representative to $w_m$; without loss of generality, we can assume that 
$d(w_\ell,w_m) \leq d(w_m,w_r)$, and thus algorithm $\mathcal{A}_{\max}^*$ outputs the solution $\bw = (w_\ell,w_m)$.  
Let $\med$ be the position of the overall median agent. 
We partition the agents into the following four sets:
$S_\ell = \{i: x_i \leq w_\ell\}$, 
$S_{\ell m} = \{i: x_i \in (w_\ell, w_m]\}$,
$S_{m r} = \{i: x_i \in (w_m,w_r)\}$, and
$S_r = \{i: x_i \geq w_r\}$. 
For any such set $S_x$, let $n_x = |S_x|$ be its size.
By the definition of $w_\ell$, $w_m$ and $w_r$, we have that 
\begin{align*}
    n_\ell, n_r \geq \frac{m-1}{2}\cdot \frac{n}{2m}  \geq \frac{(m-1)n}{4m},
\end{align*}
\begin{align*}
    n_\ell, n_r \leq n - \frac{m+1}{2}\cdot \frac{n}{2m} \leq \frac{(3m-1)n}{4m}
\end{align*}
and
\begin{align*}
    n_\ell+ n_{\ell m}, n_{\ell m}+n_r \geq \frac{m+1}{2}\cdot \frac{n}{2m}  \geq \frac{(m+1)n}{4m}.
\end{align*}
We now distinguish cases depending on the relative position of $\med$ with $w_\ell$, $w_m$ and $w_r$. 

\medskip
\noindent 
{\bf Case 1: $\med < w_\ell$.}
We have
\begin{itemize}
    \item For any $i \in S_\ell$,
        $\cost_i(\bw) = d(i,w_m) \leq d(i,\med) + d(\med,w_\ell) + d(w_\ell,w_m)$.

    \item For any $i \in S_{\ell m}$,
    $\cost_i(\bw) \leq d(w_\ell,w_m$) and $d(i,\med) \geq d(\med, w_\ell).$ 

    \item For any $i \in S_{m r}$,
    $d(i,\med) = \cost_i(\bw) + d(\med, w_\ell) \geq d(w_\ell,w_m) + d(\med,w_\ell)$.

    \item For any $i \in S_{r}$,
    $d(i,\med) = \cost_i(\bw) + d(\med, w_\ell) \geq d(w_m,w_r) + d(w_\ell,w_m) + d(\med,w_\ell) \geq 2d(w_\ell,w_m) + d(\med,w_\ell)$.
\end{itemize}
Hence, the social cost of $\bw$ is
\begin{align}
    \SC(\bw) 
    &\leq \sum_{i \in S_\ell} d(i,\med) + n_\ell d(\med,w_\ell) + (n_\ell + n_{\ell m}) d(w_\ell,w_m) + \sum_{i \in S_{m r} \cup S_r} \cost_i(\bw). \label{eq:k=2:unrestricted:SCbw-Case1}
\end{align}
For the optimal solution $\bo$, we have
\begin{align}
    \SC(\bo)
    &\geq \sum_{i \in S_\ell} d(i,\med) + \sum_{i \in S_{m r} \cup S_r} \cost_i(\bw) + (n_{\ell m} + n_{m r} + n_r) d(\med,w_\ell) \label{eq:k=2:unrestricted:SCbo-Case1a} \\
    &\geq (n_{\ell m} + n_{m r} + n_r) d(\med,w_\ell) + (n_{m r} + 2n_r) d(w_\ell,w_m). \label{eq:k=2:unrestricted:SCbo-Case1b}
\end{align}
Using \eqref{eq:k=2:unrestricted:SCbo-Case1a}, \eqref{eq:k=2:unrestricted:SCbw-Case1} becomes
\begin{align*}
    \SC(\bw) 
    &\leq \SC(\bo) + \big( n_\ell - (n_{\ell m} + n_{m r} + n_r) \big) d(\med,w_\ell) + (n_\ell + n_{\ell m}) d(w_\ell,w_m) \\
    &= \SC(\bo) + \big( n - 2(n_{\ell m} + n_{m r} + n_r) \big) d(\med,w_\ell) + (n_\ell + n_{\ell m}) d(w_\ell,w_m).
\end{align*}
Using this and \eqref{eq:k=2:unrestricted:SCbo-Case1b}, the approximation ratio is
\begin{align*}
    \frac{\SC(\bw)}{\SC(\bo)} &\leq 1 + \frac{ \big( n - 2(n_{\ell m} + n_{m r} + n_r) \big) d(\med,w_\ell) + (n_\ell + n_{\ell m}) d(w_\ell,w_m)}{(n_{\ell m} + n_{m r} + n_r) d(\med,w_\ell) + (n_{m r} + 2n_r) d(w_\ell,w_m)} \\
    &\leq 1+ \max\bigg\{\frac{n - 2(n_{\ell m} + n_{m r} + n_r)}{n_{\ell m} + n_{m r} + n_r}, \frac{n_\ell + n_{\ell m}}{n_{m r} + 2n_r} \bigg\} \\
    &= \max\bigg\{\frac{n}{n_{\ell m} + n_{m r} + n_r}-1, \frac{n + n_r}{n_{m r} + 2n_r} \bigg\}.
\end{align*}
Since $n_{m r }\geq 0$, $n_{m r} + n_r \geq \frac{(m+1)n}{4m}$ and $n_r \geq \frac{(m-1)n}{4m}$, we obtain
\begin{align*}
    \frac{\SC(\bw)}{\SC(\bo)}
    &\leq  \max\bigg\{\frac{3m}{m+1}, \frac{5m-1}{2m-2} \bigg\}
    \leq 7/2.
\end{align*}

\medskip
\noindent 
{\bf Case 2: $\med \in [w_\ell, w_m)$.}
We have
\begin{itemize}
    \item For any $i \in S_\ell$, since $\cost_i(\bw) = d(i,w_m)$, 
     $d(i,\med) = \cost_i(\bw) - d(\med,w_m) \geq d(w_\ell,\med)$.

    \item For any $i \in S_{\ell m}$,
    $\cost_i(\bw) \leq d(w_\ell,w_m) = d(w_\ell,\med) + d(\med,w_m).$ 

    \item For any $i \in S_{m r}$, since $\cost_i(\bw) = d(i,w_\ell)$,
    $d(i,\med) = \cost_i(\bw) - d(w_\ell,\med) \geq d(\med,w_m)$.

    \item For any $i \in S_{r}$, since $\cost_i(\bw) = d(i,w_\ell)$,
    $d(i,\med) = \cost_i(\bw) - d(w_\ell,\med) \geq d(w_m,w_r) + d(\med,w_m) \geq (w_\ell,w_m) + d(\med,w_m) = d(w_\ell,\med) + 2d(\med,w_m)$.
\end{itemize}
For the optimal cost, we have
\begin{align}
    \SC(\bo) 
    &\geq \sum_{i \in N} d(i,\med) \nonumber \\
    &\geq \sum_{i \in S_\ell \cup S_{m r} \cup S_r} \cost_i(\bw) - n_\ell d(\med,w_m) - (n_{m r} + n_r) d(w_\ell,\med) \label{eq:k=2:unrestricted:SCbo-Case2a} \\
    &\geq (n_\ell+n_r) d(w_\ell,\med) + (n_{m r} + 2n_r) d(\med,w_m). \label{eq:k=2:unrestricted:SCbo-Case2b}
\end{align}
Using \eqref{eq:k=2:unrestricted:SCbo-Case2a}, the social cost of $\bw$ is
\begin{align*}
    \SC(\bw) \leq \sum_{i \in N} \cost_i(\bw) \leq \SC(\bo) + (n_{\ell m} + n_{m r} + n_r) d(w_\ell,\med) + (n_\ell+n_{\ell m}) d(\med,w_m).
\end{align*}
Using this and \eqref{eq:k=2:unrestricted:SCbo-Case2b}, the approximation ratio is
\begin{align*}
    \frac{\SC(\bw)}{\SC(\bo)} 
    &\leq 1 + \frac{(n_{\ell m} + n_{m r} + n_r) d(w_\ell,\med) + (n_\ell+n_{\ell m}) d(\med,w_m)}{(n_\ell+n_r) d(w_\ell,\med) + (n_{m r} + 2n_r) d(\med,w_m)} \\
    &\leq 1 + \max\bigg\{ \frac{n_{\ell m} + n_{m r} + n_r}{n_\ell+n_r}, \frac{n_\ell+n_{\ell m}}{n_{m r} + 2n_r} \bigg\} \\
    &= 1 + \max\bigg\{ \frac{n - n_\ell}{n_\ell+n_r}, \frac{n - n_{m r} - n_r}{n_{m r} + 2n_r} \bigg\} \\
    &= \max\bigg\{ \frac{n + n_r}{n_\ell+n_r}, \frac{n + n_r}{n_{m r} + 2n_r} \bigg\}.
\end{align*}
Since $n_{m r }\geq 0$, and $n_\ell, n_r \geq \frac{(m-1)n}{4m}$, we obtain
\begin{align*}
    \frac{\SC(\bw)}{\SC(\bo)}
    &\leq \frac{5m-1}{2m-2}
    \leq 7/2.
\end{align*}

\medskip
\noindent 
{\bf Case 3: $\med \in [w_m,w_r]$.}
We have
\begin{itemize}
    \item For any $i \in S_\ell$, 
        $\cost_i(\bw) = d(i,w_m) = d(i,w_\ell) + d(w_\ell,w_m)$ and 
        $d(i,\med) = d(i,w_\ell) + d(w_\ell,w_m) + d(w_m,\med)$.

    \item For any $i \in S_{\ell m}$,
        $\cost_i(\bw) \leq d(w_\ell,w_m$) and
        $d(i,\med) \geq d(w_m,\med).$ 

    \item For any $i \in S_{m r}$,
       $\cost_i(\bw) = d(i,w_\ell) \leq d(i,\med) + d(w_m,\med) + d(w_\ell,w_m)$.

    \item For any $i \in S_r$,
    $\cost_i(\bw) = d(i,w_\ell) \leq d(i,w_r) + d(w_r,\med) + d(w_m,\med) + d(w_\ell,w_m)$ and
    $d(i,\med) = d(i,w_r) + d(\med, w_r)$. 
\end{itemize}
Hence, the social cost of $\bw$ is
\begin{align*}
    \SC(\bw) 
    &\leq \sum_{i \in S_\ell} d(i,w_\ell) + \sum_{i \in S_{mr}} d(i,\med) + \sum_{i \in S_r} d(i,w_r) \nonumber \\
    &\quad + n \cdot d(w_\ell,w_m) + (n_{m r} + n_r) d(w_m,\med) + n_r d(\med,w_r).
\end{align*}
Since $n_{m r} + n_r = n - ( n_{\ell} + n_{\ell m})$, we further have that
\begin{align}
    \SC(\bw) 
    &\leq \sum_{i \in S_\ell} d(i,w_\ell) + \sum_{i \in S_{mr}} d(i,\med) + \sum_{i \in S_r} d(i,w_r) \nonumber \\
    &\quad + n \cdot d(w_\ell,w_m) + \big( n - ( n_{\ell} + n_{\ell m}) \big) d(w_m,\med) + n_r d(\med,w_r). \label{eq:k=2:unrestricted:SCbw-Case3}
\end{align}
For the optimal solution $\bo$, we have
\begin{align*}
    \SC(\bo) 
    &\geq \sum_{i \in N} d(i,\med) \\
    &\geq \sum_{i \in S_\ell} d(i,w_\ell) + \sum_{i \in S_{mr}} d(i,\med) + \sum_{i \in S_r} d(i,w_r) \\
    &\quad + n_\ell d(w_\ell,w_m) + (n_\ell +n_{\ell m}) d(w_m,\med) + n_r d(\med,w_r).
\end{align*}
By rearranging terms, we have
\begin{align*}
    &\sum_{i \in S_\ell} d(i,w_\ell) + \sum_{i \in S_{mr}} d(i,\med) + \sum_{i \in S_r} d(i,w_r) + n_r d(\med,w_r) \\
    &\leq \SC(\bo)  - n_\ell d(w_\ell,w_m) - (n_\ell +n_{\ell m}) d(w_m,\med).
\end{align*}
Substituting this into \eqref{eq:k=2:unrestricted:SCbw-Case3}, 
we obtain
\begin{align*}
    \SC(\bw)
    &\leq \SC(\bo) + (n-n_\ell) d(w_\ell,w_m) + \big(n - 2(n_{\ell} + n_{\ell m})\big) d(w_m,\med).
\end{align*}
The approximation ratio is 
\begin{align*}
    &\frac{\SC(\bw)}{\SC(\bo)} \leq 1
    + \frac{(n-n_\ell) \cdot d(w_\ell,w_m) + \big(n - 2(n_{\ell} + n_{\ell m})\big) d(w_m,\med)}{ n_\ell d(w_\ell,w_m) + (n_\ell +n_{\ell m}) d(w_m,\med) + n_r d(\med,w_r)}.
\end{align*}
Observe now that the last expression is a monotone function in terms of $d(w_m,\med)$ and in terms of $d(\med,w_r)$. So, it attains its maximum value when either $\med = w_m$ or $\med = w_r$. Since $d(w_\ell,w_m) \leq d(w_m,w_r)$, we have
\begin{align*}
    &\frac{\SC(\bw)}{\SC(\bo)} \\
    &\leq 1 + \max\bigg\{ \frac{(n-n_\ell) \cdot d(w_\ell,w_m)}{ n_\ell d(w_\ell,w_m) + n_r d(w_m,w_r)}, \frac{(n-n_\ell) \cdot d(w_\ell,w_m) + \big(n - 2(n_{\ell} + n_{\ell m})\big) d(w_m,w_r)}{ n_\ell d(w_\ell,w_m) + (n_\ell +n_{\ell m}) d(w_m,w_r)} \bigg\} \\
    &\leq 1 + \max\bigg\{ \frac{n-n_\ell}{n_\ell+n_r}, \frac{n-n_\ell + n - 2(n_{\ell} + n_{\ell m})}{ n_\ell + n_\ell +n_{\ell m}}  \bigg\} \\
    &\leq \max\bigg\{ \frac{n+n_r}{n_\ell+n_r}, \frac{2n- n_{\ell} - n_{\ell m}}{2n_\ell + n_{\ell m}}   \bigg\} \\
    &\leq   \max\bigg\{ \frac{5m-1}{2m-2}, \frac{7m-1}{2m} \bigg\} \\
    &\leq 7/2,
\end{align*}
where the second to last inequality follows by the bounds on $n_\ell$, $n_{\ell m}$ and $n_r$.

\medskip
\noindent 
{\bf Case 4: $\med > w_r$.}
We have 
\begin{itemize}
    \item For any $i \in S_\ell$, 
        $\cost_i(\bw) = d(i,w_m) = d(i,w_\ell) + d(w_\ell,w_m)$ and 
        $d(i,\med) = d(i,w_\ell) + d(w_\ell,w_m) + d(w_m,\med)$.

    \item For any $i \in S_{\ell m}$,
        $\cost_i(\bw) \leq d(w_\ell,w_m$) and
        $d(i,\med) \geq d(w_m,\med).$ 
        
    \item For any $i \in S_{m r} \cup S_r$,
    $\cost_i(\bw) = d(i,w_\ell) \leq d(i,\med) + d(w_m,\med) + d(w_\ell,w_m).$
\end{itemize}
Hence, the social cost of $\bw$ is
\begin{align*}
    \SC(\bw) 
    &\leq \sum_{i \in S_\ell} d(i,w_\ell) + \sum_{i \in S_{mr}\cup S_r} d(i,\med) + n \cdot d(w_\ell,w_m) + (n_{m r} + n_r) d(w_m,\med).
\end{align*}
Since $n_{m r} + n_r = n - ( n_{\ell} + n_{\ell m})$, we further have that
\begin{align}
    \SC(\bw) 
    &\leq \sum_{i \in S_\ell} d(i,w_\ell) + \sum_{i \in S_{mr}\cup S_r} d(i,\med) + n \cdot d(w_\ell,w_m) + \bigg(n - ( n_{\ell} + n_{\ell m})\bigg) d(w_m,\med). \label{eq:k=2:unrestricted:SCbw-Case4}
\end{align}
For the optimal solution $\bo$, we have
\begin{align*}
    \SC(\bo) 
    &\geq \sum_{i \in N} d(i,\med) \\
    &\geq \sum_{i \in S_\ell} d(i,w_\ell) + \sum_{i \in S_{mr} \cup S_r} d(i,\med) + n_\ell d(w_\ell,w_m) + (n_\ell +n_{\ell m}) d(w_m,\med).
\end{align*}
By rearranging terms, we have
\begin{align*}
    &\sum_{i \in S_\ell} d(i,w_\ell) + \sum_{i \in S_{mr} \cup S_r} d(i,\med) \leq \SC(\bo)  - n_\ell d(w_\ell,w_m) - (n_\ell +n_{\ell m}) d(w_m,\med).
\end{align*}
Substituting this into \eqref{eq:k=2:unrestricted:SCbw-Case4}, 
we obtain
\begin{align*}
    \SC(\bw)
    &\leq \SC(\bo) + (n-n_\ell) d(w_\ell,w_m) + \bigg(n - 2(n_{\ell} + n_{\ell m})\bigg) d(w_m,\med).
\end{align*}
The approximation ratio is 
\begin{align*}
    &\frac{\SC(\bw)}{\SC(\bo)} \leq 1
    + \frac{(n-n_\ell) \cdot d(w_\ell,w_m) + \bigg(n - 2(n_{\ell} + n_{\ell m})\bigg) d(w_m,\med)}{ n_\ell d(w_\ell,w_m) + (n_\ell +n_{\ell m}) d(w_m,\med)}.
\end{align*}
As in Case 3, since $d(w_\ell,w_m) \leq d(w_m,w_r) \leq d(w_m,\med)$, we have that 
\begin{align*}
\frac{\SC(\bw)}{\SC(\bo)} 
&\leq \frac{2n- n_{\ell} - n_{\ell m}}{2n_\ell + n_{\ell m}} \\
&\leq \frac{7m-1}{2m} \\
&\leq 7/2,
\end{align*}
where the second to last inequality follows by the bounds on $n_\ell$ and $n_{\ell m}$. 
\end{proof}

\section{Mechanisms For $k\geq 3$ Facilities}\label{sec:kgeq3:sp}
We now turn our attention to the case of more than two facilities. We show that the approximation ratio of strategyproof mechanisms is between $3-2/k$ and $3+2/k$ in the sum-variant, and between $2k$ and $2(k+1)$ in the max-variant. Our upper bounds follow by mechanisms that are generalizations of the mechanism $\mathcal{M}_m(\theta,\ell,r)$ which was used for the case of $k=2$ facilities in previous sections: They first choose the group representative locations according to the ordering of the agents in the groups, and then choose the final facility locations according to the ordering of the representatives. Because of this structure, the proof of Theorem~\ref{thm:k=2:general-sp} can be extended to show that these mechanisms are also strategyproof, and the proof of Lemma~\ref{lem:k=2-wlog-property} can be extended to show that it suffices to consider instances with symmetric groups to bound their approximation ratio. 

\subsection{Sum-Variant} \label{sec:kgeq3:sp:sum}
We start with the sum-variant and first show the upper bound of $3+2/k$. Our mechanism $\mathcal{M}_k^{\text{sum}}$ works as follows: 
In Phase 1, for each group $g$, it selects the position of the (leftmost) median agent therein as $g$'s representative.
In Phase 2, the mechanism opens the facilities at the $\lceil \frac{\ell m}{k+1}\rceil$-leftmost representative location, for any $\ell \in [k]$; observe that $\lceil \frac{\ell m}{k+1}\rceil < \lceil \frac{(\ell+1)m}{k+1}\rceil$, and hence the computed solution is feasible. See Mechanism~\ref{mech:kgeq3:sum:mechanism} for a description using pseudocode.  

\renewcommand{\algorithmcfname}{Mechanism}
\begin{algorithm}[h]
\SetNoFillComment
\caption{$\mathcal{M}_k^{\text{sum}}$}
\label{mech:kgeq3:sum:mechanism}
{\bf Input:} Instance $I = (\bx, G, k)$ \;
{\bf Output:} Facility locations $\bw = (w_\ell)_{\ell \in [k]}$\;
$\mathtt{Rep} \gets \varnothing$\;
\For{each group $g \in G$}
{
    $i_g \gets \lceil n_g/2 \rceil$-leftmost agent in $g$\;
    $r_g \gets x_{i_g}$\;
    $\mathtt{Rep} \gets \mathtt{Rep} \cup \{r_g\}$\;
}
\For{each $\ell \in [k]$}
{
    $w_\ell \gets \lceil \frac{\ell m}{k+1}\rceil$-leftmost location in $\mathtt{Rep}$\;
}
\end{algorithm}

We now present the bound on the approximation ratio of the mechanism. 

\begin{theorem}\label{thm:kgeq3:sum:upper}
In the sum-variant, for $k \geq 3$ facilities, the approximation ratio of mechanism $\mathcal{M}_k^{\text{sum}}$ is at most $3+2/k$.
\end{theorem}

\begin{proof}
First, recall that it suffices to consider only instances with $m$ symmetric groups, i.e., $n_g = n/m$ for any group $g$. We also assume without loss of generality that $m$ is a multiple of $k+1$.\footnote{Observe that for any instance $I$ where $m$ is not a multiple of $k+1$, there is another instance $J$ consisting of $k+1$ copies of each group in $I$, for which the solution computed by the mechanism is clearly the same (since the representatives are chosen according to equidistant quantiles). In addition, the position of the median agent is the same in both instances since each agent has been copied the same number of times.} Let $\bw = (w_1,\ldots,w_k)$ be the solution computed by the mechanism, and denote by $\med$ the position of the overall median agent. 
For simplicity, also let $w_0 = -\infty$ and $w_{k+1} = +\infty$. 
Define the sets $L = \{\ell \in [k]: w_\ell \leq \med \}$ and $R = \{\ell \in [k]: w_\ell > \med \}$.
Clearly, $L$ and $R$ define a partition of $[k]$.
In addition, for any $\ell \in \{2,\ldots,k\}$, if $\ell \in L$, then $\ell-1 \in L$, and, similarly, for any $\ell \in \{1,\ldots,k-1\}$, if $\ell \in R$, then $\ell+1 \in R$. 
Note that it might be the case that $L$ or $R$ is empty in case $\med$ is to the left or to the right of all representative locations that are included in the solution $\bw$, respectively. 

For any $\ell \in L$, let $S_\ell = \{i \in g: r_g \in (w_{\ell-1}, w_\ell] \wedge x_i \leq r_g \}$ 
be the set of agents in groups $g$ with representative $r_g \in (w_{\ell-1}, w_\ell]$ with positions weakly to the left of $r_g$. 
Similarly, for any $\ell \in R$, let $S_\ell = \{i \in g: r_g \in [w_\ell, w_{\ell+1}) \wedge x_i \geq r_g \}$ 
be the set of agents in groups $g$ with representative $r_g \in [w_\ell,w_{\ell+1})$ with positions weakly to the right of $r_g$. By definition, these sets are disjoint. Since for any group $g$, $r_g$ is the position of the median agent in $g$, and since there are at least $\frac{m}{k+1}$ group representatives between any two representatives in $\bw$ (and also to the left of $w_1$, i.e., between $w_0$ and $w_1$, and to the right of $w_k$, i.e., between $w_k$ and $w_{k+1}$), we have that $n_\ell := |S_\ell| \geq \frac{m}{k+1}\cdot \frac{n}{2m} = \frac{n}{2(k+1)}$ for any $\ell \in [k]$. 

Now, using the triangle inequality, the social cost of $\bw$ is
\begin{align}
    \SC(\bw) &= \sum_{i\in N} \sum_{\ell=1}^k d(i,w_\ell) \nonumber \\
    &\leq \sum_{i\in N} \sum_{\ell=1}^k \bigg( d(i,\med) + d(\med,w_\ell) \bigg) \nonumber \\ 
    &= k \cdot \sum_{i \in N} d(i,\med) + n \cdot \sum_{\ell=1}^k d(\med,w_\ell). \label{eq:kgeq3:sum:alg}
\end{align}
For the optimal social cost, since an agent in $S_\ell$ needs to travel the distance $d(\med,w_\ell)$ to reach from her position to $\med$ (note that $x_i \leq r_g \leq w_\ell \leq \med$ for $\ell \in L$ and $x_i \geq r_g \geq w_\ell \geq \med$ for $\ell \in R$), we have
\begin{align}
    \SC(\bo) &\geq k \cdot \sum_{i \in N} d(i,\med) \label{eq:kgeq3:sum:opt-1}\\
    &\geq k \cdot \sum_{\ell=1}^k \sum_{i \in S_\ell} d(\med,w_\ell) \nonumber \\
    &= k \cdot \sum_{\ell=1}^k n_\ell \cdot  d(\med,w_\ell) \nonumber \\
    &\geq \frac{kn}{2(k+1)} \cdot \sum_{\ell=1}^k d(\med,w_\ell). \label{eq:kgeq3:sum:opt-2}
\end{align}
Using \eqref{eq:kgeq3:sum:opt-1} and \eqref{eq:kgeq3:sum:opt-2}, \eqref{eq:kgeq3:sum:alg} gives us that
\begin{align*}
    \SC(\bw) \leq \SC(\bo) +  \frac{2(k+1)}{k} \cdot \SC(\bo) = \left( 3 + \frac{2}{k} \right) \cdot \SC(\bo),
\end{align*}
and the proof is complete.
\end{proof}

Next, we show a lower bound of $3-2/k$ on the approximation ratio of any algorithm. Combined with the above upper bound, these two results leave a small gap for any $k \geq 3$, which, however, becomes nearly $0$ when the number $k$ of facilities is rather large.  

\begin{theorem}\label{thm:kgeq3:sum:lower}
In the sum-variant, for $k \geq 3$ facilities, the approximation ratio of any algorithm is at least $3-2/k$.
\end{theorem}

\begin{proof}
Without loss of generality, assume that in a group with one agent at $0$ and another agent at $1$, the algorithm selects $0$ as the group's representative. Then, consider an instance with $k$ groups such that the first group has an agent at $0$ and another agent at $1$, whereas the remaining $k-1$ groups have both agents at $1$. The algorithm has no choice but to open a facility at $0$ and $k-1$ facilities at $1$ for a social cost of $\SC(\bw) = 3k-2$. However, the optimal solution is to open all $k$ facilities at $1$ for a social cost of $\SC(\bo) = k$, and the claim follows.
\end{proof}

\subsection{Max-Variant} \label{sec:kgeq3:sp:max}
We now consider the max-variant and again start by showing the upper bound. Our mechanism 
$\mathcal{M}_k^{\text{max}}$ works as follows: 
In Phase 1, for each group $g$, it selects the position of the (leftmost) median agent therein as $g$'s representative.
In Phase 2, the mechanism opens the facilities at the $k$ most central representative locations, i.e.,  at the
$(\lceil \frac{m}{2} \rceil + \ell)$-leftmost representative for 
$\ell \in \{-(\lceil \frac{k}{2} \rceil-1), \ldots, 0, \ldots, \lfloor \frac{k}{2} \rfloor\}$.
See Mechanism~\ref{mech:kgeq3:max:mechanism} for a description using pseudocode.  

\renewcommand{\algorithmcfname}{Mechanism}
\begin{algorithm}[h]
\SetNoFillComment
\caption{$\mathcal{M}_k^{\text{max}}$}
\label{mech:kgeq3:max:mechanism}
{\bf Input:} Instance $I = (\bx, G, k)$ \;
{\bf Output:} Facility locations $\bw = (w_\ell)_{\ell \in [k]}$\;
$\mathtt{Rep} \gets \varnothing$\;
\For{each group $g \in G$}
{
    $i_g \gets \lceil n_g/2 \rceil$-leftmost agent in $g$\;
    $r_g \gets x_{i_g}$\;
    $\mathtt{Rep} \gets \mathtt{Rep} \cup \{r_g\}$\;
}
\For{each $\ell \in [k]
$}
{
        $w_\ell \gets \left(\left\lceil \frac{m}{2} \right\rceil + \ell - \left\lceil \frac{k}{2} \right\rceil\right)$-leftmost location in $\mathtt{Rep}$\;

}
\end{algorithm}

\begin{theorem} \label{thm:kgeq3:max:upper}
In the max-variant, for $k \geq 3$ facilities, the approximation ratio of mechanism $\mathcal{M}_k^{\text{max}}$ is at most $2(k+1)$.
\end{theorem}

\begin{proof}
Let $\bw = (w_1,\ldots,w_k)$ be the solution computed by the mechanism, and denote by $\med$ the position of the overall median agent. 
Let $L= \{i: x_i \leq w_1 \}$ and $R=\{i: x_i \geq w_k\}$. By definition, we have that $|L|,|R| \geq \left( \lfloor \frac{m-k}{2} \rfloor + 1\right) \frac{n}{2m}$, since the representative of each group is the position of the median agent therein, and there are at least $\lfloor \frac{m-k}{2} \rfloor$ groups whose representatives do not participate in the solution and are (weakly) to the left of $w_1$ and (weakly) to the right of $w_k$.
The individual cost of an agent $i$ is
\begin{align*}
    \cost_i(\bw) 
    &=  
    \begin{cases}
        d(i,w_k) &\text{ if $i \in L$} \\
        \max\{d(i,w_1),d(i,w_k) \}  &\text{ if $i \notin L\cup R$} \\
        d(i,w_1)  &\text{ if $i \in R$}
    \end{cases} \\
    &\leq 
    \begin{cases}
        d(i,w_1) + d(w_1,w_k) &\text{ if $i \in L$} \\
        d(w_1,w_k)  &\text{ if $i \notin L\cup R$} \\
        d(i,w_k) + d(w_1,w_k)  &\text{ if $i \in R$}
    \end{cases}
\end{align*}
So, the social cost of $\bw$ is
\begin{align}
    \SC(\bw) \leq \sum_{i \in L} d(i,w_1) + \sum_{i \in R} d(i,w_k) + n\cdot d(w_1,w_k). \label{eq:kgeq:max:alg}
\end{align}
For the optimal cost, we have
\begin{align*}
    \SC(\bo) \geq \sum_{i \in N} d(i,\med). 
\end{align*}

We now switch between cases depending on the relative position of $\med$ with $w_1$ and $w_k$.

\medskip
\noindent
{\bf Case 1: $\med < w_1$.} (The case $\med > w_k$ can be handled similarly.)
We have that $d(i,\med) = d(i,w_k) + d(\med, w_k)$ for each agent $i \in R$. 
So, the optimal social cost can be lower-bounded as
\begin{align*}
    \SC(\bo) \geq \sum_{i \in L} d(i,\med) + \sum_{i \in R} d(i,w_k) +  |R| \cdot d(\med,w_k). 
\end{align*}
For the social cost of $\bw$, using the triangle inequality, and since $|L|\leq n$ and $d(\med,w_1) + d(w_1,w_k) = d(\med,w_k)$, \eqref{eq:kgeq:max:alg} gives us that
\begin{align*}
    \SC(\bw) 
    &\leq \sum_{i \in L} d(i,w_1) + \sum_{i \in R} d(i,w_k) + n\cdot d(w_1,w_k) \\
    &\leq \sum_{i \in L} d(i,\med) + |L| d(\med,w_1) + \sum_{i \in R} d(i,w_k) + n\cdot d(w_1,w_k) \\
    &\leq \sum_{i \in L} d(i,\med) + \sum_{i \in R} d(i,w_k) + n \cdot d(\med,w_k)  \\
    &\leq \SC(\bo) + (n-|R|) \cdot d(\med,w_k).
\end{align*}
Hence, the approximation ratio is
\begin{align*}
    \frac{\SC(\bw)}{\SC(\bo)} \leq 1 + \frac{n-|R|}{|R|} = \frac{n}{|R|} \leq \frac{2m}{\lfloor \frac{m-k}{2} \rfloor+1}.
\end{align*}
The last expression is non-increasing in terms of $m$ for any value of $k$, and thus, depending on the parity of $m$ and $k$, it acquires its maximum value either for $m=k$, which leads to a bound of $2k$, or for $m=k+1$, which leads to a bound of $2(k+1)$. 

\medskip
\noindent
{\bf Case 2: $\med \in [w_1,w_k]$.}
For the optimal cost, since $d(w_1,\med) + d(\med,w_k) = d(w_1,w_k)$, we have
\begin{align*}
    \SC(\bo) &\geq \sum_{i \in L} d(i,\med) + \sum_{i \in R} d(i,\med) \\
    &= \sum_{i \in L} d(i,w_1) + |L|d(w_1,\med) + \sum_{i \in R} d(i,w_k) + |R| \cdot d(\med,w_k) \\
    &\geq \sum_{i \in L} d(i,w_1) + \sum_{i \in R} d(i,w_k) + \min\{|L|,|R|\} \cdot d(w_1,w_k).
\end{align*}
Using this, \eqref{eq:kgeq:max:alg} gives us that 
\begin{align*}
    \SC(\bw) \leq \SC(\bo) + (n - \min\{|L|,|R|\} ) \cdot d(w_1,w_k).
\end{align*}
Hence, the approximation ratio is at most
\begin{align*}
    \frac{\SC(\bw)}{\SC(\bo)} \leq 1 + \frac{n-\min\{|L|,|R|\} }{\min\{|L|,|R|\} } = \frac{n}{\min\{|L|,|R|\} } \leq \frac{2m}{\lfloor \frac{m-k}{2} \rfloor+1}.
\end{align*}
As in Case 1, the last expression acquires its maximum value of $2(k+1)$ for $m=k+1$. 
\end{proof}

Finally, we show a lower bound of $2k$ for any unrestricted algorithm, which holds even for $k=2$. Overall, this implies that $4$ is the best possible approximation ratio over unrestricted algorithms for $k=2$ (but not over strategyproof mechanisms, for which the bound is $9/2=4.5$), and mechanism $\mathcal{M}_k^{\text{max}}$ is asymptotically best possible for large $k$. 

\begin{theorem}\label{thm:kgeq3:max:lower}
In the max-variant, for $k\geq 2$ facilities, the approximation ratio of any algorithm is at least $2k$.
\end{theorem}

\begin{proof}
Let $g$ be a group with one agent at $0$ and one agent at $1$. Due to symmetry, suppose without loss of generality that the algorithm selects $0$ as the representative of $g$. Consider an instance with $g$ and another $k-1$ groups consisting of two agents at $1$. Clearly, $1$ is the representative of each of the last $k-1$ groups, and the facilities are therefore opened at $0$ and $1$. Then, $\SC(\bw) = 2k$. However, the optimal solution is to open all $k$ facilities at $1$ for a social cost of $1$, leading to an approximation ratio of $2k$.
\end{proof}

\section{Conclusion and Open Questions}\label{sec:conclusion}
In this work, we considered a truthful distributed facility location model where the goal is to open $k \geq 2$ facilities at locations chosen from the set of agent positions which are locally reported within their groups. We showed bounds on the best possible approximation ratio of strategyproof mechanisms in both the sum-variant (where the individual cost of an agent is the distance to all facilities) and the max-variant (where the individual cost of an agent is the distance to the farthest facility). Our bounds are tight for $k=2$ facilities in both variants, but there are small gaps for $k\geq 3$ facilities, which appear to be quite challenging to close.  

There are several other future directions for extending our results and model. Our investigation of deterministic strategyproof mechanisms turned out to be quite rich and led to small and mostly tight bounds, but it would also be interesting to investigate whether randomized mechanisms can lead to significantly better approximation bounds. In addition, one can consider alternative modeling assumptions, by changing the social objective or the individual cost definition. For example, one can study alternative objectives such as the sum-of-max or max-of-sum, which have been considered in several other distributed social choice models~\citep{AnshelevichFV2022distributed,Voudouris2023tight,FilosRatsikasKVZ24,Xu2025distributed,AmanatidisAJV2025group}. The individual cost could alternatively be defined as the distance from the $q$-th closest facility for $q \in [k]$ (note that the max-variant corresponds to the case $q=k$), similarly to previous work on the distortion of metric multiwinner voting~\citep{caragiannis2022multi}, peer selection~\citep{Cembrano2025peer}, and sortition~\citep{Ebadian2022sortition}. Finally, another direction could be to consider models with (possibly inaccurate) predictions about, for example, the optimal facility location or the agent positions, following the learning-augmented framework~\citep{Agrawal2024learning,BalkanskiGS2024randomized,christodoulou2024advice,FangFLNV2025predictions,deligkas2025outliers}.

\bibliographystyle{plainnat}
\bibliography{references}

\end{document}